\documentclass[journal,draftcls,onecolumn,12pt,twoside]{IEEEtran}
\usepackage{algpseudocode}
\usepackage{algorithm}
\usepackage{graphicx}
\usepackage{amsmath}
\newcounter{mytempeqncnt}
\usepackage{amssymb}
\usepackage{amsthm}
\usepackage{epstopdf}
\usepackage{cite}
\usepackage{mathtools,lipsum}
\usepackage{mathrsfs}
\usepackage{upgreek}
\usepackage{bm}     
\usepackage{amsfonts}
\usepackage{textcomp}
\usepackage{caption}
\usepackage[numbers]{natbib}
\newtheorem{monTh}{Theorem}
\newtheorem{monCr}{Corollary}

\ifCLASSINFOpdf
 \else
\fi

\begin{document}
\title{Effective Capacity Analysis of H-ARQ Assisted Cooperative Communication Systems}
\author{Mokhtar~Bouteggui,~\IEEEmembership{}
        Fatiha~Merazka,~\IEEEmembership{}
        and~Gunes~Karabulut~Kurt~\IEEEmembership{Senior Member}
\thanks{M. Bouteggui and F. Merazka are with the Department of Telecommunications, Electrical Engineering Faculty, USTHB University, 16111, Algiers, Algeria, e-mail: \{mbouteggui, fmerazka\}@usthb.dz.}
\thanks{G. Karabulut Kurt is with the Department of Communications and Electronics Engineering,
Istanbul Technical University, 34469, Istanbul, Turkey, e-mail: gkurt@itu.edu.tr.}
\thanks{Manuscript received April 19, 2005; revised August 26, 2015.}}
\markboth{IEEE TRANSACTIONS ON WIRELESS COMMUNICATION,~Vol.~14, No.~8, August~2015}%
{Shell IEEE TRANSACTIONS ON WIRELESS COMMUNICATION,~Vol.~14, No.~8, August~2015}
\maketitle
\begin{abstract}
In this paper, the effective capacity of cooperative communication (CC) systems with hybrid Automatic repeat request (HARQ) is derived.
The derived expressions are valid for any channel distribution and with any arbitrary number of retransmissions by the source and relay for both HARQ-repetition redundancy (RR) and HARQ-incremental redundancy (IR) over asymmetric channels. As an example, we use the derived EC expression over Rayleigh fading channels. Several results are obtained for a low rate and signal-to-noise ratio (SNR). We can see that the EC attends its maximum value with a small number of retransmissions. As expected when the relay-destination channel has low SNR, it is better than the relay does not participate especially when we assign a large number of transitions at the relay. For high data rates and strict quality of service (QoS) constraints, it is better to increase the number of relay transmissions. Finally, when we increase the number of source retransmissions, the effective capacity improves even for low values.
\end{abstract}
\section{Introduction}
\IEEEPARstart Wireless communication systems have experienced several revolutions undergone several revolutionary developments in the last decade. However, these systems suffer from unreliable channels, which introduce the need for new wireless transmission protocols enabling higher data rates and improved communication reliability. When the received packet is detected to be in error, packet retransmission is often requested and used for more reliable and efficient communication over fading channels \cite{caire_throughput_2001}. This scheme is called automatic repeat request (ARQ). The efficiency of ARQ can be improved by combining new and previously received information instead of discarding them, through the use of the so-called Hybrid-ARQ (HARQ) schemes \cite{larsson_analysis_2013}.
HARQ schemes may be classified into two categories: i) Chase combining HARQ (CC-HARQ), ii) Incremental redundancy HARQ (IR-HARQ) \cite{larsson_throughput_2014}\cite{larsson_throughput_2016}.
Another line classification for HARQ is, i) Truncated-HARQ, for limiting the maximum number of retransmissions, or ii) Persistent-HARQ by continuous retransmissions until error-free decoding of the packet \cite{larsson_throughput_2014}\cite{yang_truncated-arq_2016}.
To improve communication reliability, HARQ schemes can be jointly used with cooperative communication (CC) techniques. In cellular systems, and due to the advantage broadcast nature of the wireless channel, CC techniques can enable high reliability by combating the performance degrading effects of the wireless fading channels \cite{laneman_cooperative_2004}\cite{su_cooperative_2008}. Through the use of relay nodes that aid communication between the sender and the destination, the diversity order of the system can be improved. 
The performance evaluations of CC systems frequently focus on error or outage analysis. However, for some data services, such as video streaming, video conferencing, and online gaming over wireless networks, the delay is critically an important parameter for quality of service (QoS) guarantees. On the other hand, due to traffic and channel randomness, the reliable transmission may not be provided all the time. Therefore, depending on the type of data transmission, delay-violation probability and buffer overflow concerns become critical for the sender.
So, it is essential to model a wireless channel in terms of QoS metrics such as data rate, delay, and delay-violation probability. However, the existing channel models do not explicitly characterize a wireless channel in terms of such QoS metrics.\\
An alternative performance metric for queuing systems, is the notion of effective capacity \cite{dapeng_wu_effective_2003}, which is the dual of the effective bandwidth \cite{chen-shang_chang_effective_1995}.
The effective capacity provides the maximum sustainable arrival rate for the buffer under certain QoS constraints. The concept of effective capacity has gained notable attention, and it has been investigated in several transmission scenarios.
For cc, starting with a single source, relay, and destination, the authors in \cite{harsini_effective_2012} analyze and optimize the effective capacity for decode and forward (DF) relay system, where both, source and relay are equipped with adaptive modulation and coding (AMC), in conjunction with a cooperative ARQ protocol over time-correlated fading channels.
In \cite{hu_blocklength-limited_2016}, the authors consider DF relay system with quasi-static fading channel. Both of these works investigate, blocklength-limited and the effective capacity of relaying systems. For the channel, the case of only the average channel state information (CSI) is available at the sender and the case of perfect CSI is considered.
In \cite{phan_optimal_2016}, a buffer-aided relaying network was considered where
both source and relay employ their buffer. The authors assume that there is no direct link between the source
and destination, adaptive link selection relaying schemes are proposed for fixed and adaptive power allocation at
source and relay.
In \cite{yuli_yang_relay_2013} the relay nodes are considered having different modulation capabilities.
The effective capacity is derived first, then an effective capacity-based relay selection criterion is proposed.
In \cite{hu_qos-constrained_2016}, a system with a single source, destination, and multi-relay with DF protocol is considered. Four cooperative ARQ protocols are considered where only the average CSI is available at the source and relays. A maximal-ratio combining (MRC) is assumed to be used at the destination during the retransmissions.
In \cite{qiao_fixed_2017}, the authors consider a buffer-aided diamond relay systems, with full-duplex DF relays. The buffers at source and relay nodes are considered of infinite size, with statistical queuing constraints for each one. The maximum effective capacity is obtained for the conventional relay selection protocol, (DF, time division-DF, and broadcast-DF). Also, a selection forwarding policy is proposed to maximize the statistical delay exponent at the source.
In \cite{karatza_effective_2018}, a novel analytical framework for effective capacity is introduced, a multi-source multi-destination system with AF relaying protocol. The relay nodes are assumed without buffer, with the presence of co-channel interference.
Exact analytical expressions and tight bounds expressions for several linear precoding techniques are derived. The effective capacity maximizing problem is also addressed.
The main contributions of our work are summarized below. To the best of our knowledge, there is no prior work or effective capacity expression for CC with retransmission schemes with any arbitrary number of transmissions at the source and the relay using packet combining at both relay and destination over asymmetric channels for any distribution. Also, there are no works that use the recurrence relation approach for CC.

\begin{itemize}
\item Two strategies are proposed.
\item Outage probabilities are derived and verified via simulation results.
  \item The effective capacity for CC using lossless ARQ relay and destination is derived.
  \item The effective capacity for CC for any number of retransmissions from the source and the relay, for both HARQ-RR and HARQ-IR, over asymmetric channels for any channel distribution is derived.
\end{itemize}
\subsection{Assumptions}
A half duplex system is considered. The feedbacks are assumed to be error-free and latency-free. The channel coefficients are assumed to be constant during the packet but change every packet, and the noise is additive white Gaussian noise (AWGN). Perfect CSI is assumed at the receivers. At the $k$-th and $l$-th transmissions, we assume  $P_{uv,k}=1-Q_{uv,k}$, and $P_{srd,k;l}=1-Q_{srd,k;l}$. Both relay and destination use packets combating.
\vspace{-0.8cm}
\subsection{Notation}
For notation simplification, we denote by $s,r$ and $d$ the source, relay and destination respectively. We let $u$ and $v$ where $u\in \{s,r\}$, $v\in\{r,d\}$ and $u\neq v$. All random variables are iid. We use the notation $X_i\sim E(\mu_i)$ to denote that a random variable $X_i$ follows exponential distribution with parameter $\mu_i$. For $t>0$, the probability density function (pdf) and cumulative distribution function (cdf) of $X_i$ are $f_{X_i}(t)$ and $F_{X_i}(t)$ are given by \cite[3.19a and 3.19b]{miller_probability_2004}, respectively.
Let $H_i=X_i+\alpha_i$ be a shifted exponential random variable with parameter $\mu_i$ and the shifted parameter $\alpha_i$. We use the notation $H_i\sim ShE(\mu_i,\alpha_i)$ to denote a shifted exponential random variable with parameters $\mu_i$ and $\alpha_i$. The pdf and cdf of $H_i$ are $f_{H_i}(t)=f_{X_i}(t-\alpha_i)$ and $F_{H_i}(t)=F_{X_i}(t-\alpha_i)$, respectively.
The lower incomplete gamma function is $\gamma(a,x)$ and the upper incomplete gamma function is $\Gamma(a,x)$ are given by \cite[8.350.1 and 8.350.2]{gradshtein_table_2015}, respectively.
The Mellin transform of a function $f(t)$ denoted by $\mathcal{M}_s\{f(t)\}(s)$ is defined by \cite[2.1]{mathai_h-function:_2010}.
The inverse Mellin transform is given by
\cite[2.2]{mathai_h-function:_2010}.
Let $\large{H}^{m,n}_{p,q}$ be the generalized upper incomplete Fox’s $H$ function given by \cite[10]{chelli_performance_2014}.
\section{System model}
\subsection{System description}
We consider a CC system composed of a single source, relay, and destination. The relay uses selective DF (S-DF) as the relaying protocol \cite{su_cooperative_2008}. The transmission is completed in two orthogonal time slots, in the first time slot the source sends the symbol $x$ to the destination which is also received by the relay. In the second time slot, if the relay can decode $x$, then it forwards $x$ to the relay, otherwise the relay remains idle. At the destination, a maximum ratio combiner is used to decode the transmitted symbols.
\subsection{Channel model}
Let $y_{sd}$ and $y_{sr}$ be the received signals at the destination and the relay from the source in the first time slot and $y_{rd}$ the received signal from the relay at the destination in the second time slot. These signals are given by
\begin{align}
  y_{sd}=& \sqrt{P_s}h_{sd}x+n_{sd} \\
  y_{sr}=& \sqrt{P_s}h_{sr}x+n_{sr} \\
  y_{rd}=& \sqrt{P_r}h_{rd}x+n_{rd}
\end{align}
where $P_s$ and $P_r$ are the transmitted power at the source (s) and the relay (r), $h_{uv}$ are the channel coefficients between $u$ and $v$ where $u\in \{s,r\}$, $v\in\{r,d\}$ and $u\neq v$. The noise between $u$ and $v$ is $n_{uv}$ is a complex AWGN with zero mean and variance $N_{uv}$.
\section{Effective capacity}
\subsection{QoS Exponent}
The exponential decay rate of the delay-bound QoS violation
probabilities indicated by the QoS exponent $\theta$ \cite{cheng-shang_chang_stability_1994}, is given by
\begin{equation}
  -\lim_{q_{th}\to \infty} \frac{\ln(\mathbb{P}\{q(\infty)>q_{th}\})}{q_{th}}=\theta
\end{equation}
where $q_{th}$ is a threshold indicating the queue length bound.
Larger $\theta$, corresponds to more stringent QoS requirements, which implies that the system cannot tolerate any delay.
On the other hand, when $\theta$ is smaller, corresponds to looser QoS requirement, which implies that the system can tolerate an arbitrarily long delay.
\subsection{Recurrence Relation Approach}
When we use the recurrence relation approach, the considered retransmission schemes have a set of multiple communication modes and multiple decoded packets \cite{larsson_effective_2016}. For example, in \cite{larsson_effective_2016}, the authors considered a two users NC-ARQ.
For that, a set $S$ with $L=3$ communication modes are used and the maximum number of decoded packets is $v_{max}=2$.
For multiple communication modes and multiple decoded packets is $2$-user NC-ARQ in \cite{larsson_effective_2016}. For more details, we refer the reader to see \cite{larsson_effective_2016}.
\begin{monTh}\cite[Cor 1]{larsson_effective_2016}\label{Thm1}
For retransmission scheme and transmission rate $R$, with $k \to \infty$, let $S$ be a set with $L$ communication modes, $v\in\{0,1,\cdots v_{max}\}$ communication packets per transmission and $P_{v\tilde{s}s}$ the transition probabilities from communication mode $\tilde{s}$ to $s$.
For each communication mode $s$, the effective capacity is given by
\begin{equation}\label{Eq2}
C_{\text{eff}}=-\frac{\ln(\lambda_+)}{\theta},
\end{equation}
 where $\lambda_+=\max\{|\lambda_1|,|\lambda_2|,\cdots,|\lambda_{L}|\}$ is the spectral radius of the block companion matrix $\mathbf{A}$ given by
 \begin{equation}\label{Eq3}
\mathbf{A}=\begin{bmatrix}
  \alpha_{11} & \alpha_{21} & \cdots & \alpha_{\tilde{L}1}\\
  \alpha_{12} & \alpha_{22} & \cdots & \alpha_{\tilde{L}2}\\
  \vdots & \vdots & \ddots & \vdots\\
  \alpha_{1L} & \alpha_{2L} & \cdots & \alpha_{\tilde{L}L}\\
\end{bmatrix}
 \end{equation}
and
$\{\lambda_1,\lambda_2,\cdots,\lambda_{L}\}$ are the eigenvalues of $\mathbf{A}$ and
\begin{equation}\label{AAA}
      \alpha_{\tilde{s}s}=\sum_{v=0}^{v_{max}}P_{v\tilde{s}s}e^{-\theta vR}
\end{equation}
\end{monTh}
\section{Effective capacity for CC with retransmission schemes}
\subsection{Strategy I}
In this strategy, lossless ARQ is used by the relay and destination. Since we use ARQ, each transmitted packet by the source is independently examined by the relay and destination. Similarly, each packet transmitted by the relay is independently examined by the destination. Also, for lossless ARQ each packet will be transmitted until being correctly decoded by the destination. For asymmetric channels, the probability that both the relay and definition fail to decode the transmitted packet by the source using ARQ are
\begin{align}
  Q_{sd}=&\mathbb{P}\left\{\log_2\left(1+\Gamma_{sd}g_{sd,i}\right)\leq R\right\}  \\
  Q_{sr}=&\mathbb{P}\left\{\log_2\left(1+\Gamma_{sr}g_{sr,i}\right)\leq R\right\}
\end{align}
where $\Gamma_{sd}=\frac{P_s}{N_{sd}}$ and $\Gamma_{sr}=\frac{P_s}{N_{sr}}$ are the received signal to noise ratios (SNR) at destination and relay,  respectively and $g_{sd,i}=|h_{sd,i}|^2$  and $g_{sr,i}=|h_{sr,i}|^2$ are the channel gain at the $i$-th transmission between source-destination and source relay. The channel gains between $u$ and $v$ are random variables with a set of parameters denoted by $\Delta_{uv}$. For example, Rayleigh fading channels gain are exponential random variables with parameter $\Delta_{uv}=(1/\delta^2_{uv})$ and for Nakagami-$m$ fading, the channels gain are gamma distribution with parameters $\Delta_{uv}=(m_{uv},1/\delta^2_{uv})$.
The probability that definition failed to decode the transmitted packet after transmission by the relay using ARQ is given by
\begin{align}
  Q_{rd}=&\mathbb{P}\left\{\log_2\left(1+\Gamma_{rd}g_{rd,i}\right)\leq R\right\}
\end{align}
where $\Gamma_{rd}=\frac{P_r}{N_{rd}}$ is the received SNR at destination $g_{rd,i}=|h_{rd,i}|^2$ is the channel gain at the $i$-th transmission between relay-destination.
In the case where all channels are symmetric, we have
\begin{equation}
  Q_{uv}=\mathbb{P} \left\{\log_2\left(1+\Gamma g_{i}\right)\leq R\right\}
\end{equation}
where $\Gamma=\frac{P}{N}$ and $g_{i}$ are random variables with a set of parameters $\Delta$ for $u\in\{s,r\}$ $v\in\{r,d\}$ such that $u\neq v$. This means that the transmit power at the source and relay is $P$, the noise variance in all channels is $N$ and all channels gains have the same distribution with the same set of parameters $\Delta$. Also, let $P_{uv}=1-Q_{uv}$ be the probability of an error-free packet decoded.
\begin{monCr}\label{Cr1}
The effective capacity of a CC system using strategy I, is given by Theorem \ref{Thm1}
for $k \to \infty$ with
\begin{equation}\label{CC_Mat1}
\mathbf{A}=\begin{bmatrix}
  Q_{sd}Q_{sr}+P_{sd}e^{-\theta R} &P_{rd}e^{-\theta R} \\
  P_{sr}Q_{sd} & Q_{rd}
\end{bmatrix}.
\end{equation}
\end{monCr}
\begin{proof}
We consider two-mode operations with $S = \{1, 2\}$. In mode $s=1$, the source sends a packet to the destination. We enter mode $s = 2$ when a packet intended for destination is correctly decoded only by the relay but not by the destination with probability $Q_{sd}P_{sr}$. Now, in mode $s=2$, the relay keeps on sending packets until the destination decodes correctly. If the destination succeeds to decode correctly, we go back to mode $s=1$ with probability $P_{rd}$, otherwise we stay in mode $s=2$ with probability $Q_{rd}$. Using equation (\ref{AAA}), the matrices $\mathbf{A}$ in Corollary \ref{Cr1} is readily identified.
\end{proof}
\subsection{Strategy II}
In this strategy, the source and the relay use truncated HARQ i.e the number of transmissions of packets from the source and the relay is $M$ and $N$, respectively. Also, when the source sends, both relay and destination use HARQ. At the last attempt from the source or the relay, if the destination fails to decode the current packet, this packet will be dropped and the source sends a new packet. For each packet, the transmission starts by the source. If the current packet is correctly decoded by the destination, both the source and the relay receive an acknowledgment (ACK) from the destination, and the source sends a new packet. If the destination fails to decode the current packet correctly then the destination feeds back negative acknowledgment (NACK). Here, while the number of retransmissions by the source is less than or equal to $M$ and both the relay and the destination fail to decode the transmitted packet, the source keeps on sending. If the number of retransmissions exceeds $M$ and both the relay and the destination do not decode correctly the transmitted packet, the source drops this packet and sends a new one. If the relay decodes the current packet, the source stops sending and the relay starts the retransmission until the destination decodes the packet or the maximum number of retransmissions is reached. If the number of retransmissions by the relay exceeds $N$, the relay stops the transmission and the source sends a new packet.

First, we consider the use of HARQ-RR.
Let $k_1\in \{1,2,\cdots,M\}$ the $k_1$-th transmission by the source where both the relay and the destination do not decode correctly the current packet.
The probabilities that both the relay and the destination fail to decode the transmitted packet by the source in the $k_1$ transmission using HARQ-RR are
\begin{align}\label{RR_Sym1}
  Q_{sd,k_1}=&\mathbb{P}\left\{\log_2\left(1+\sum_{i=1}^{k_1}\Gamma_{sd}g_{sd,i}\right)\leq R\right\}  \\ \label{RR_Sym2}
  Q_{sr,k_1}=&\mathbb{P}\left\{\log_2\left(1+\sum_{i=1}^{k_1}\Gamma_{sr}g_{sr,i}\right)\leq R\right\}
\end{align}
Let $l\in\{1,2,\cdots,M\}$ the $l$-th transmission where only the relay decodes the transmitted packet but not the destination and let $k_2\in \{1,2,\cdots,N\}$ the $k_2$-th transmission by the relay to the destination. We assume that the probability of an error-free packet decoded at
the $l$-th transmission is $P_{sd,l}=1-Q_{sd,l}$.
The probability that the destination fails to decode the transmitted packet after $l$ transmissions from the source and $k_2$ transmissions from the relay using HARQ-RR is given by
\begin{align}\label{RR_Asym3}
  Q_{srd,l;k_2}=&\mathbb{P}\left\{\log_2\left(1+\sum_{i=1}^{l}\Gamma_{sd}g_{sd,i}+\sum_{i=1}^{k_2}\Gamma_{rd}g_{rd,i}\right)\leq R\right\}.
\end{align}
The probability of an error-free packet decoded after $l$ transmissions from the source and $k_2$ from the relay is $P_{srd,l;k_2}=1-Q_{srd,l;k_2}$.
In the case where the channels are symmetries, for HARQ-RR we have
\begin{align}\label{RR_Sym4}
  Q_{uv,k_1}&=\mathbb{P}\left\{\log_2\left(1+\sum_{i=1}^{k_1}\Gamma g_i\right)\leq R\right\}\\ \label{RR_Sym5}
  Q_{srd,l;k_2}&=\mathbb{P}\left\{\log_2\left(1+\sum_{i=1}^{l+k_2}\Gamma g_i\right)\leq R\right\}
  \end{align}
%

We also consider the use of HARQ-IR, where the probabilities that both the relay and the destination fail to decode the transmitted packet by the source in the $k_1$ transmissions using HARQ-IR are
\begin{align}\label{IR_Sym1}
  Q_{sd,k_1}=&\mathbb{P}\left\{\sum_{i=1}^{k_1}\log_2\left(1+\Gamma_{sd}g_{sd,i}\right)\leq R\right\} \\ \label{IR_Sym2}
  Q_{sr,k_1}=&\mathbb{P}\left\{\sum_{i=1}^{k_1}\log_2\left(1+\Gamma_{sr}g_{sr,i}\right)\leq R\right\}
\end{align}
The probability that destination fails to decode the transmitted packet after $l$ transmissions by the source and $k_2$ transmissions by the relay using HARQ-IR is given by
\begin{align}\label{IR_Asym3}
  Q_{srd,l;k_2}=&\mathbb{P}\Bigg\{\sum_{i=1}^{l}\log_2\left(1+\Gamma_{sd}g_{sd,i}\right)+\sum_{i=1}^{k_2}\log_2\left(1+\Gamma_{rd}g_{rd,i}\right)\leq R\Bigg\}
\end{align}
In the case where the channels are symmetries, for HARQ-IR we have
\begin{align}\label{IR_Sym4}
  Q_{uv,k_1}&=\mathbb{P}\left\{\sum_{i=1}^{k_1}\log_2\left(1+\Gamma g_i\right)\leq R\right\}\\ \label{IR_Sym5}
  Q_{srd,l;k_2}&=\mathbb{P}\left\{\sum_{i=1}^{l+k_2}\log_2\left(1+\Gamma g_i\right)\leq R\right\}
  \end{align}
\begin{monCr}\label{Cr2}
For a CC using strategy II with $M$ transmissions by the source and $N$ transmissions by the relay. Let $S=S_s\cup S_r$ be a set with $|S|=L=M\times (N+1)$ communication modes, where $S_s=\{s^{'}_1,s^{'}_2,\cdots s^{'}_M\}$ are the set of communication modes that the source sends and $S_r=\{s=s^{'}_l+i\}$ for all $i\in\{1,2\cdots,N\} $ and $l\in\{1,2,\cdots,M\}$ are the set of communication modes that the relay sends. The number of communicated packets is $v\in\{0,1\}$. The effective capacity is given by Theorem \ref{Thm1} for $k\to \infty$ with $\mathbf{A}$ defined by
\begin{description}
  \item [(1)] $s^{'}_1=1$ the mode where the source sends first time.
  \item [(2)] For the $l$-th transmission from source the equivalent mode is $s^{'}_l=s^{'}_{l-1}+N+1$, for $l\in \{2,\cdots M\}$.
  \item [(3)] At the $l$-th attempts from the source, if only the relay decodes the packet correctly, the relay sends packet for the destination in modes $s=s^{'}_l+1$ up to $s=s^{'}_l+N$.
  \item [(4)] If $M=1$, we stay in mode $s^{'}=1$ with probability $Q_{sd,1}Q_{sr,1}$ or $P_{sd,1}$. We go to mode $s=2$ with probability $Q_{sd,1}P_{sr,1}$.
  \item [(5)] If $M>1$, we stay in mode $s^{'}_1=1$ with probability $P_{sd,1}$. While $1<l< M$, we go from mode $s^{'}_{l-1}$ to mode $s^{'}_l$ with a probability $Q_{sd,l-1}Q_{sr,l-1}$ or to mode $s^{'}_1=1$ with probability $P_{sd,l-1}$ . If $l=M$, first
       we go from mode $s^{'}_{M-1}$ to mode $s^{'}_M$ with probability $Q_{sd,M-1}Q_{sr,M-1}$ and to mode $s^{'}_1=1$ with probability $P_{sd,M-1}$, then we go from mode $s_M^{'}$ to mode $s^{'}_1=1$ with probability $Q_{sd,M}Q_{sr,M}$ or $P_{sd,M}$.
  %
  \item [(6)] If $M>1$, for $l\in\{1,2,\cdots,M\}$, we go from mode $s^{'}_l$ to mode $s=s^{'}_l+1$ with probability $Q_{sd,l}P_{sr,l}$. The relay also uses HARQ.
  \item [(7)] For each $l\in\{1,2,\cdots,M\}$, first we put $i=1$. While $i<N$ we go from mode $s=s^{'}_l+i$ to mode $s=s^{'}_l+i+1$ with probability $Q_{srd,l;i}$ or to mode $s=1$ when $D$ decodes correctly with probability $P_{srd,l;i}$ then add one to $i$ ($i=i+1$). If $i=N$, we go from mode $s=s^{'}_l+N$ to mode $s^{'}_1=1$ with probability $Q_{srd,l;N}$ or $P_{srd,l;N}$.
      \item [(8)] All other transitions are with probabilities $0$.
\end{description}
For HARQ-RR $Q_{sd,i}, Q_{sr,i}$ and $Q_{srd,l;k}$ are given by equations (\ref{RR_Sym1}),(\ref{RR_Sym2}) and (\ref{RR_Asym3}), respectively for asymmetric channels. For symmetric channels $Q_{sd,i}, Q_{sr,i}$ are given by equation (\ref{RR_Sym4}) and  $Q_{srd,l;k}$ by equation (\ref{RR_Sym5}).

For HARQ-IR,  $Q_{sd,i}, Q_{sr,i}$ and $Q_{srd,l;k}$ are given by equations (\ref{IR_Sym1}),(\ref{IR_Sym2}) and (\ref{IR_Asym3}), respectively for asymmetric channels. For symmetric channels $Q_{sd,i}, Q_{sr,i}$ are given by equation (\ref{IR_Sym4}) and  $Q_{srd,l;k}$ by equation (\ref{IR_Sym5}).

Finally, $ P_{uv,i}=1-Q_{uv,i}$ and $P_{srd,i}=1-Q_{srd,i}$.
\end{monCr}
\begin{proof}
First, for each transmission either from the source or the relay, we associate a mode. Now, the way of modes that we use is as follows. After each transmission from the source, the next $N$ modes are associate with the relay to send. Therefore, one mode for one source transmission and $N$ modes for the relay which is equivalent to $N+1$ modes. Finally, we multiply by the number of source transmissions, which is $M$, we have $M\times(N+1)$ communication modes. For (1), we use mode $s^{'}_1=1$ to denote the mode where the source sends the new packet. For (2) and (3) are the modes where the source and the relay send, respectively. For (4), since $M=1$, the source sends only one time, so either no one decodes this packet (relay or destination) with probability $Q_{sd,1}Q_{sr,1}$ therefore this packet will be dropped to send a new one, or the destination decodes correctly with probability $P_{sd,1}$, in which, we stay in mode $s^{'}_1=1$ to send a new packet. We go to mode $s=2$ only when the relay decodes correctly but not the destination with probability $Q_{sd,1}P_{sr,1}$. For (5), since the source sends $M$ times, so we stay in mode $s^{'}_1$ if a packet is correctly decoded by the destination with probability $P_{sd,1}$. If at the $l-1$-th transmission from the source either the relay or destination does not be able to decode correctly we go to the $l$-th transmission with probability $Q_{sd,l-1}Q_{sr,l-1}$ since both the relay and the destination use HARQ and they have received $l-1$ copies of the same packet.
At the last transmission from the source, we go back to mode $s^{'}_1=1$ in the case where no one decodes the transmitted packet with probability $Q_{sd,M}Q_{sd,M}$ after $M$ transmissions, or the destination decodes the packet with probability $P_{sd,M}$ after $M$ transmissions.
For (6), if only the relay decodes the packet after combining $l$ transmissions from the source, we go to mode $s^{'}_l+1$ where the relay sends, with probability $Q_{sd,l}P_{sr,l}$.
Finally for (7), if only the relay decodes the current packet after $l$ transmissions from the source, the relay starts sending. Note that the destination also has received $l$ packets from the source. So, while the number of transmissions of the relay is less than $N$, the relay sends to the destination, where this latter combines the previous $l$ packets from the source (direct link) and the packets received from the relay. At the $i$-th transmission from the relay, where $i<N$, if the destination fails to decode the transmitted packet we go to the next mode with probability $Q_{srd,l;i}$ or to mode $s^{'}_1=1$ if the destination decodes correctly with probability $P_{srd,l;i}$. At the last transmission from the relay, we go back to mode $s^{'}_1=1$ in the case where the destination fails to decode the packet correctly with probability $Q_{srd,l;N}$, or decodes correctly with probability $P_{srd,l;N}$. All other transition probabilities are zero.
\end{proof}
\begin{itemize}
\item Example 1: $M=N=1$

For $M=N=1$, we have $|S|=1\times(1+1)=2$ communication modes, $S_s=\{1\}$ and $S_r=\{2\}$. Using (2), since $M=1$, the source sends only on mode where $l=1$ which is $s^{'}_1=1$. Using (3), the relay sends on mode $s=s^{'}_1+1=2$. Using (4), we stay in mode $s^{'}_1=1$ with probability $Q_{sd,1}Q_{sr,1}$ or $P_{sd,1}$ or we go to mode $Q_{sd,1}P_{sr,1}$. Using (7), for $i=N=1$, we go from mode $s=2$ to mode $s^{'}_1=1$ with probabilities $Q_{srd,1;1}$ or $P_{srd,1;1}$.
      \begin{equation}
      \mathbf{A}=\begin{bmatrix}
        Q_{sd,1}Q_{sr,1} + P_{sd,1}e^{-\theta R}  & Q_{srd,1;1} + P_{srd,1;1}e^{-\theta R}  \\
        Q_{sd,1}P_{sr,1}  &0
      \end{bmatrix}
      \end{equation}
\item Example 2: $M=1$ and $N=2$

For $M=1$ and $N=2$, we have $|S|=1\times(2+1)=3$, $S_s=\{1\}$ and $S_r=\{2,3\}$ . Using (2), the source sends only in the case where $l=M=1$, and therefore mode $s^{'}_1=1$. Using (3), since $l$ equals only one i.e $l=1$, the relay sends in modes $s^{'}_1+1=2$ and $s^{'}_1+2=3$. Using (4), since $M=1$ we stay in mode $s^{'}_1=1$ with probilities $Q_{sd,1}Q_{sr,1}$ or $P_{sd,1}$. We go to mode $s=2$ with probability $Q_{sd,1}P_{sr,1}$. Using (7), since $l=1$, for $i=1$ we go from mode $s=2$ to mode $s=3$ with probability $Q_{srd,1;1}$ or to mode $s^{'}_1=1$ with probability $P_{srd,1;1}$. For $i=N=2$, we go from mode $s=3$ to mode $s^{'}_1=1$ with probabilities $Q_{srd,1;2}$ or $P_{srd,1;2}$. The matrix $\mathbf{A}$ is given in (\ref{A_1_2}).

      \begin{equation}\label{A_1_2}
      \mathbf{A}=
      \begin{bmatrix}
        Q_{sd,1}Q_{sr,1} + P_{sd,1}e^{-\theta R}&P_{srd,1;1}e^{-\theta R}&Q_{srd,1;2} + P_{srd,1;2}e^{-\theta R}\\
        Q_{sd,1}P_{sr,1}  &0 &0\\
        0&Q_{srd,1;1} &0
      \end{bmatrix}
      \end{equation}
\item Example 3: $M=2$ and $N=1$

For $M=2$ and $N=2$, we have $|S|=2\times(1+1)=4$, $S_s=\{1,3\}$ and $S_r=\{2,4\}$. Using (2), for $l=1$ the source sends in modes $s_1^{'}=1$ and for $l=2$ the source sends in modes $s_2^{'}=3$. Using (3), for $l=1$ the relay sends in mode $s=2$. For $l=2$ the relay sends in mode $s=4$. Using (5), we stay in mode $s_1^{'}=1$ with probability $P_{sd,1}$. For $l=M=2$, we go from $s_1^{'}=1$ to mode $s^{'}_2=3$ with probability $Q_{sd,1}Q_{sr,1}$ and mode $s_1^{'}=1$ with probability $P_{sd,1}$ (here since $M-1=1$ so it is equivalent to that when we stay in mode $s_1^{'}$). We go from mode $s^{'}_2=3$ to mode $s_1^{'}=1$ with probabilities $Q_{sd,2}Q_{sr,2}$ or $P_{sd,2}$. Using (6), for $l=1$ we go from mode $s^{'}_1=1$ to mode $s=s^{'}_1+1=2$ with probability $Q_{sd,1}P_{sr,1}$. For $l=2$, we go from mode $s^{'}_2=3$ to mode $s=s^{'}_2+1=4$ with probability $Q_{sd,2}P_{sr,2}$. Using (7), for $l=1$ and $i=N=1$ we go from  mode $s=2$ to mode $s_1^{'}=1$ with probabilities $Q_{srd,1;1}$ or $P_{srd,1;1}$. For $l=2$ and $i=N=1$ we go from mode $s=4$ to mode $s_1^{'}=1$ with probabilities $Q_{srd,2;1}$ or $P_{srd,2;1}$. The matrix $\mathbf{A}$ is given in (\ref{A_2_1}).

          \begin{equation}\label{A_2_1}
            \begin{bmatrix}
              P_{sd,1}e^{-\theta R} & Q_{srd,1;1}+P_{srd,1;1}e^{-\theta R}& Q_{sd,2}Q_{sr,2}+P_{sd,2}e^{-\theta R} & Q_{srd,2;1}+P_{srd,2;1}e^{-\theta R} \\
              Q_{sd,1}P_{sr,1} & 0 & 0 & 0  \\
              Q_{sd,1}Q_{sr,1} & 0 & 0 & 0  \\
              0& 0 & Q_{sd,2}P_{sr,2} & 0  \\
            \end{bmatrix}
          \end{equation}

\item Example 4: $M=2$ and $N=2$

For $M=N=2$ we have $|S|=2\times(2+1)=6$, $S_s=\{1,4\}$ and $S_r=\{2,3,5,6\}$. Using (2), for $l=1$ the source sends in mode $s_1^{'}=1$ and for $l=2$ the source sends in mode $s_2^{'}=4$.
 Using (3), for $l=1$ the relay sends in modes $s=2$ and $s=3$ and for $l=2$ the relay sends in modes $s=5$ and $s=6$.
 Using (5), we stay in mode $s_1^{'}=1$ with probability $P_{sd,1}$. For $l=M=2$, we go from $s_1^{'}$ to mode $s^{'}_2$ with probability $Q_{sd,1}Q_{sr,1}$ and mode $s=1$ with probability $P_{sd,1}$ (here since $M-1=1$ so it is equivalent to that when we stay in mode $s_1^{'}$). We go from mode $s^{'}_2=4$ to mode $s=1$ with probabilities $Q_{sd,2}Q_{sr,2}$ or $P_{sd,2}$.
 Using (6), for $l=1$ we go from mode $s^{'}_1=1$ to mode $s=s^{'}_1+1=2$ with probability $Q_{sd,1}P_{sr,1}$. For $l=2$ we go from mode $s^{'}_2=4$ to mode $s=s^{'}_2+1=5$ with probability $Q_{sd,2}P_{sr,2}$.
 Using (7), for $l=1$ and $i=1$ we go from mode $s=2$ to mode $s=3$ with probability $Q_{srd,1;1}$ or to mode $s^{'}_1=1$ with probability $P_{srd,1;1}$. For $i=2=N$ we go from mode $s=3$ to mode $s^{'}_1=1$ with probabilities $Q_{srd,1;2}$ or $P_{srd,1;2}$.
 For $l=2$ and $i=1$ we go from mode $s=5$ to mode $s=6$ with probability $Q_{srd,2;1}$ or to mode $s^{'}_1=1$ with probability $P_{srd,2;1}$. For $i=2=N$ we go from mode $s=6$ to mode $s^{'}_1=1$ with probabilities $Q_{srd,2;2}$ or $P_{srd,2;2}$. The matrix $\mathbf{A}$ is given in (\ref{A_2_2}).

          \begin{equation}\label{A_2_2}
          \small
            \begin{bmatrix}
              P_{sd,1}e^{-\theta R} & P_{srd,1;1}e^{-\theta R} & Q_{srd,1;2}+P_{srd,1;2}e^{-\theta R} & Q_{sd,2}Q_{sr,2}+P_{sd,2}e^{-\theta R} & P_{srd,2;1}e^{-\theta R} & Q_{srd,2;2}+P_{srd,2;2}e^{-\theta R} \\
              Q_{sd,1}P_{sr,1} & 0 & 0 & 0 & 0 & 0 \\
              0 & Q_{srd,1;1} & 0 & 0 & 0 & 0 \\
              Q_{sd,1}Q_{sr,1} & 0 & 0 & 0 & 0 & 0 \\
              0 & 0 & 0 & Q_{sd,2}P_{sr,2} & 0 & 0 \\
              0 & 0 & 0 & 0 & Q_{srd,2;1} & 0
              \end{bmatrix}
          \end{equation}
 \end{itemize}         
\normalsize
\section{Application over Rayleigh fading channel}
For channel coefficient $h_{uv}$, we consider a Rayleigh flat fading channel where the gain $g_{uv}=|h_{uv}|^2$, is exponentially distributed with parameter $\Delta_{uv}=(1/\delta_{uv}^2)$ i.e $g_{uv}\sim E(1/\delta_{uv})$. Also we note that $\Gamma_{uv}g_{uv}$ is exponential random variable with parameter $1/(\Gamma_{uv}\delta_{uv}^2)$.
\\
\vspace{-0.8cm}
\subsection{Outage Probability of Strategy I}
Since in strategy I we use ARQ, all the transmissions are independent, we have
\begin{align}
Q_{sd}=&\mathbb{P}\left\{\log_2\left(1+\Gamma_{sd}g_{sd,i}\right)\leq R\right\}\nonumber  \\
      =&\mathbb{P}\left\{\Gamma_{sd}g_{sd,i}\leq 2^R-1\right\}\nonumber\\ \label{Ou_ARQ}
      =&1-e^{\Theta_{sd}}
\end{align}
where $\Theta_{sd}=\frac{2^R-1}{\Gamma_{sd}\delta_{sd}^2}$. Similarly, $Q_{sr}$ and $Q_{rd}$ are given by (\ref{Ou_ARQ}) where $\Theta_{sr}=\frac{2^R-1}{\Gamma_{sr}\delta_{sr}^2}$ and $\Theta_{rd}=\frac{2^R-1}{\Gamma_{rd}\delta_{rd}^2}$.
For symmetric channels, we have $Q_{uv}=\frac{2^R-1}{\Gamma\delta}$ where $1/\delta$ is the parameter of all channels.
\subsection{Outage Probabilities of  Strategy II using HARQ-RR}
\begin{monCr}\label{Lemma1}
Let $X_i$ be a exponential random variable with parameter $\mu_i$. Let $Y_i$ be a random variable defined by the sum of $k_i$ exponential random variables $X_j$ with the same parameter $\mu_i$, i.e $X_j\sim E(\mu_i)$ for $j\in \{1,2,\cdots,k_i\}$
\begin{equation}\label{Lem1}
  Y_i=\sum_{j=1}^{k_i}X_j.
\end{equation}
The random variable $Y_i$ has an Erlang distribution $Erl(k_i,\mu_i)$ with pdf $f_{Y_i}(t)$ and cdf $F_{Y_i}(t)$ given by
\begin{align}\label{Lem2}
    f_{Y_i}(t)&=\frac{\mu_i^{k_i}t^{k_i-1}}{(k_i-1)!}e^{-\mu_it}\\ \label{Lem3}
    F_{Y_i}(t)&=\frac{1}{(k_i-1)!}\gamma(k_i,t\mu_i)
\end{align}
for $t>0$. Let $Z$ be a random variable defined by the sum of two Erlang random variables $Y_1\sim Erl(k_1,\mu_1)$ and $Y_2\sim Erl(k_2,\mu_2)$, $Z=Y_1+Y_2$.
The pdf and cdf of $Z$ are $f_Z(t)$ and $F_Z(t)$, respectively given by
\begin{align}
  f_Z(t)=&\sum_{i=1}^{2}\mu_i^{k_i}e^{-\mu_it}\sum_{j=1}^{k_i}\frac{(-1)^{k_i-j}}{(j-1)!}t^{j-1}\times\sum_{\begin{matrix}
  n_1+n_2=k_i-j \\
  n_i=0
\end{matrix}}^{}
\prod_{\begin{matrix}
  l=1 \\
  l\neq i
\end{matrix}}^{2}
\binom{k_l+n_l-1}{n_l}\frac
{\mu_l^{k_l}}
{(\mu_l-\mu_i)^{k_l+n_l}}\\ \label{Lem5}
F_Z(t)=&\sum_{i=1}^{2}\mu_i^{k_i}\sum_{j=1}^{k_i}\frac{(-1)^{k_i-j}}{(j-1)!}
\mu_i^{-j}\gamma(j,t\mu_i)
   \times\sum_{\begin{matrix}
  n_1+n_2=k_i-j \\
  n_i=0
\end{matrix}}^{}
\prod_{\begin{matrix}
  l=1\\
  l\neq i
\end{matrix}}^{2}
\binom{k_l+n_l-1}{n_l}\frac
{\mu_l^{k_l}}
{(\mu_l-\mu_i)^{k_l+n_l}}
\end{align}
\end{monCr}
\begin{proof}
  See Appendix \ref{App2}
\end{proof}
First, we consider HARQ-RR over asymmetric channels. The probability that the destination fails to decode packets at the $k_1$-th transmission from the source is given by
\begin{align}\label{HARQRR}
  Q_{sd,k_1}&=\mathbb{P} \left\{\log_2\left(1+
  \sum_{i=1}^{k_1}\Gamma_{sd}g_{sd,i}\right)
  \leq R\right\}\nonumber\\
  &=\mathbb{P}\left\{\sum_{i=1}^{k_1}\Gamma_{sd}g_{sd,i}\leq \Theta\right\}\nonumber\\
  &\stackrel{(a)}{=}\mathbb{P}\left\{Y_{sd}\leq \Theta\right\}\nonumber\\
  &\stackrel{(b)}{=}\int_{0}^{\Theta}f_{Y_{sd}}(t) dt\nonumber\\
  &\stackrel{(c)}{=}F_Y{_{sd}}(\Theta)\nonumber\\
  &=\frac{1}{(k_1-1)!}\gamma(k_1,\frac{\Theta}{\Gamma_{sd}\delta^2_{sd}})
\end{align}
where $\Theta=2^R-1$ and in step $(a)$ we define $Y_{sd}=\sum_{i=1}^{k}\Gamma_{sd}h_{sd,i}$. Since $\Gamma_{uv}g_{uv,i}\sim E(1/\Gamma_{uv}\delta^2_{uv})$. Therefore in step $(b)$ by using Corollary \ref{Lemma1}, $Y_{sd}$ is a random variable following Erlang distribution,  i.e
$Y_{sd}\sim Erl(k_1,1/\Gamma_{uv}\delta^2_{uv})$. Finally in step $(c)$, we use cdf given in (\ref{Lem3}). Similarly, for $Q_{sr,k_1}$ we have
\begin{equation}
    Q_{sr,k_1}=\frac{1}{(k_1-1)!}\gamma(k_1,\frac{\Theta}{\Gamma_{sr}\delta^2_{sr}}).
\end{equation}
In that case only the relay decodes the transmitted packet after $k_1$ transmissions from the source and the relay sends $k_2$ times to the destination. The probability that the destination fails to decode the transmitted packet correctly after $k_1$ transmissions from the source and $k_2$ transmissions from the relay is
\begin{align}
  Q_{srd,k_1;k_2}&=\mathbb{P}\left\{\log_2\left(1+
  \sum_{i=1}^{k_1}\Gamma_{sd}g_{sd,i}
  +
  \sum_{i=1}^{k_2}\Gamma_{rd}g_{rd,i}\right)
  \leq R\right\}\nonumber\\ \nonumber
  &=\mathbb{P}\left\{\sum_{i=1}^{k_1}\Gamma_{sd}g_{sd,i}
  +
  \sum_{i=1}^{k_2}\Gamma_{rd}g_{rd,i}
  \leq \Theta\right\}\\ \nonumber
  &\stackrel{(a)}{=}\mathbb{P}\{Z\leq \Theta\}\nonumber\\
  &\stackrel{(b)}{=}\int_{0}^{\Theta}f_Z(t) dt\nonumber\\
  &\stackrel{(c)}{=}F_Z(\Theta)\nonumber\\
  &=\sum_{i=1}^{2}\mu_i^{k_i}\sum_{j=1}^{k_i}\frac{(-1)^{k_i-j}}{(j-1)!}
\mu_i^{-j}\gamma(j,t\mu_i)
   \sum_{\begin{matrix}
  n_1+n_2=k_i-j \\
  n_i=0
\end{matrix}}^{}
\prod_{\begin{matrix}
  l=1\\
  l\neq i
\end{matrix}}^{2}
\binom{k_l+n_l-1}{n_l}\frac
{\mu_l^{k_l}}
{(\mu_l-\mu_i)^{k_l+n_l}}
\end{align}
where $ \Theta=2^R-1$ and in step $(a)$ we define  $Z=\sum_{i=1}^{k_1}\Gamma_{sd}g_{sd,i}+\sum_{i=1}^{k_2}\Gamma_{rd}g_{rd,i}$. In step $(b)$ using  Corollary \ref{Lemma1}, is clear that $Z$ is the sum of two random variables with Erlang distribution i.e $\sum_{i=1}^{k_1}\Gamma_{sd}g_{sd,i}\sim Erl(k_1,\mu_1)$ and $\sum_{i=1}^{k_2}\Gamma_{rd}g_{rd,i}\sim Erl(k_1,\mu_2)$, where $\mu_1=1/\Gamma_{sd}\delta^2_{sd}$ and $\mu_2=1/\Gamma_{rd}\delta^2_{rd}$. Finally in step
(c) we use cdf given in (\ref{Lem5}). In the case when the channels are symmetric i.e the same transmit power $P$, noise variance $N$ and channel parameter $1/\delta^2$ we use cdf given in (\ref{Lem3}) to get $Q_{sd,k_1}$ and $Q_{srd,k_1;k_2}$ which are given (\ref{Farss1}) and (\ref{Farss2}) by and $Q_{sr,k_1}=Q_{sd,k_1}$.
\begin{align}\label{Farss1}
  Q_{sd,k_1}=&\frac{1}{(k_1-1)!}\gamma(k_1,\frac{N\Theta}{P\delta^2})\\ \label{Farss2}
  Q_{srd,k_1;k_2}=&\frac{1}{(k_1+k_2-1)!}\gamma(k_1+k_2,\frac{N\Theta}{P\delta^2}).
\end{align}

\subsection{Outage Probability of Strategy II using HARQ-IR}
In \cite{yilmaz_product_2009} using a generalization of the upper incomplete Fox’s H function, an exact pdf and cdf of the product of
shifted exponential variables are proposed. In \cite{makki_performance_2016} for  Rayleigh fading channels, the outage probabilities are derived for both HARQ-RR and IR \footnote{In \cite{makki_performance_2016}, the scheme model is
the total number of retranslations for both source and relay was $m$ and in our study, each one has its number of retransmissions.
In \cite[Theorem 3]{makki_performance_2016}, the outage probabilities for both HARQ-RR and IR are derived. For HARQ-RR, a Laplace
transform is used and the derived expression is similar to the sum of exponential random variables with different parameters. For HARQ-IR the derivation was a direct consequence from \cite{yilmaz_product_2009}.
The difference is that first, the noise variance is assumed to be one in all channels, and in our work, we take into account the noise of all channels. Also, a different transmit power is assumed at each transmission and we use the same power in all retransmissions from the source and the same power in all retransmissions from the relay.}

\begin{monCr}\label{Lemma2}
Let $H_i$ be a shifted exponential random variable with parameter $\mu_i$ and the shifted parameter $\alpha_i$, i.e, $ShE(\mu_i,\alpha_i)$. Let $Z_i$ be a random variable defined by the product of $k_i$ shifted
exponential random variables $H_j$ with the same
parameter $\mu_i$ and the shifted parameter $\alpha_i$, i.e,
$H_j\sim ShE(\mu_i,\alpha_i)$ for $j\in\{1,2,\cdots,k_i\}$ and $Z_i$ defined by
\begin{equation}
  Z_i=\prod_{j=1}^{k_i}H_j
\end{equation}
The pdf and cdf of $Z_i$ are $f_{Z_i}(z)$ and $F_{Z_i}(z)$, respectively given by
\begin{align}\label{Product_pdf}
  f_{Z_i}(z)&=\mu_i^{k_i}e^{k_i\alpha_i \mu_i}\times \large{H}^{k_i,0}_{0,k_i}\left( \mu_i^{k_i}z \left| \begin{array}{cc}- \\ \Psi_1\cdots \Psi_{k_i} \end{array} \right. \right)\\ \label{Product_cdf}
    F_{Z_i}(z)&=e^{k_i\alpha_i \mu_i}\times\large{H}^{k_i,1}_{1,k_i+1}\left( \mu_i^{k_i}z \left| \begin{array}{cc}(1,1,0) \\ \Psi_1,\cdots ,\Psi_{k_i},(0,1,0) \end{array} \right. \right)
\end{align}
where $\large{H}^{m,n}_{p,q}$ is the generalized upper incomplete Fox’s H function and $\Psi$ defined as
\begin{equation}\label{Psi}
  \Psi_l=(0,1,\alpha_i \mu_i) \emph{ for } l=1,\cdots,k_i.
\end{equation}
Let $Z_1=\prod_{n=1}^{k_1}H_n$ be the product of the $k_1$ shifted exponential random variables $H_n$ with the same parameter $\mu_1$  and the shifted parameter $\alpha_1$, i.e $H_n\sim ShE(\mu_1,\alpha_1), \emph{ for }all~n\in\{1,2,\cdots,k_1\}$. Similarly, let $Z_2=\prod_{m=1}^{k_2}H_m$ where $H_m\sim ShE(\mu_2,\alpha_2), \emph{ for } all ~ m\in\{1,2,\cdots,k_2\}$. Let $Z$ be a random variable defined by

\begin{equation}
  Z=Z_1\times Z_2=\prod_{n=1}^{k_1}H_n\prod_{m=1}^{k_2}H_m
\end{equation}

The pdf and cdf of $Z$ are $f_{Z}(z)$ and $F_{Z}(z)$ given by (\ref{pdf_Prod_shift_Diff}) and (\ref{cdf_Prod_shift_Diff}), respectively.
where
\begin{align*}
  \Psi_l=&(0,1,\alpha_1 \mu_1) \emph{ for } l=1,\cdots,k_1,  \\
  \Phi_l=&(0,1,\alpha_2 \mu_2) \emph{ for } l=1,\cdots,k_2.
\end{align*}

\begin{align}\label{pdf_Prod_shift_Diff}
  f_{Z}(z)&=\mu_1^{k_1} \mu_2^{k_2} e^{k_1\alpha_1 \mu_1+k_2\alpha_2 \mu_2}\large{H}^{k_1+k_2,0}_{0,k_1+k_2}\left(\mu_1^{k_1} \mu_2^{k_2}z \left| \begin{array}{cc}- \\ \Psi_1\cdots \Psi_{k_1},\Phi_{1}\cdots \Phi_{k_2} \end{array} \right. \right)\\ \label{cdf_Prod_shift_Diff}
    F_{Z}(z)&= e^{k_1\alpha_1 \mu_1+k_2\alpha_2 \mu_2}
   \large{H}^{k_1+k_2,1}_{1,k_1+k_2+1}\left( \mu_1^{k_1} \mu_2^{k_2}z \left| \begin{array}{cc}(1,1,0) \\ \Psi_1\cdots \Psi_{k_1},\Phi_{1}\cdots \Phi_{k_2},(0,1,0) \end{array} \right. \right)
\end{align}
\end{monCr}
\begin{proof}
  See Appendix \ref{App3}
\end{proof}
For HARQ-IR over asymmetric channels. The probability that the destination fails to decode a packet at the $k_1$-th transmission from the source is given
\begin{align}
  Q_{sd,k_1}&=\mathbb{P}\left\{\sum_{i=1}^{k_1}\log_2\left(1+\Gamma_{sd}g_{sd,i}\right)\leq R\right\}\nonumber \\
  &=\mathbb{P}\left\{\log_2\left(\prod_{i=1}^{k_1}(1+\Gamma_{sd}g_{sd,i})\leq R\right)\right\}\nonumber\\
  &=\mathbb{P}\left\{\prod_{i=1}^{k_1}(1+\Gamma_{sd}g_{sd,i})\leq 2^R\right\}\nonumber\\
  &\stackrel{(a)}{=}\mathbb{P}\left\{Y_{sd}\leq 2^R\right\}\nonumber\\
  &\stackrel{(b)}{=}\int_{0}^{2^R}f_{Y_{sd}}(t) dt\nonumber\\
  &\stackrel{(c)}{=}F_Y{_{sd}}(2^R)\nonumber\\ \label{Qsdk11}
&=e^{\frac{k_1}{\Gamma_{sd}\delta^2_{sd}}}\times\large{H}^{k_1,1}_{1,k_1+1}\left( \left(\frac{1}{\Gamma_{sd}\delta^2_{sd}}
\right)^{k_1}2^R \left| \begin{array}{cc}(1,1,0) \\ (1,1, \frac{1}{\Gamma_{sd}\delta^2_{sd}}
)_1,\cdots ,(1,1, \frac{1}{\Gamma_{sd}\delta^2_{sd}}
)_{k_1},(0,1,0) \end{array} \right. \right)
\end{align}
\begin{align}\label{Qsdk22}
Q_{sr,k_1}&=e^{\frac{k_1}{\Gamma_{sr}\delta^2_{sr}}}\times\large{H}^{k_1,1}_{1,k_1+1}\left( \left(\frac{1}{\Gamma_{sr}\delta^2_{sr}}
\right)^{k_1}2^R \left| \begin{array}{cc}(1,1,0) \\ (1,1, \frac{1}{\Gamma_{sr}\delta^2_{sr}}
)_1,\cdots ,(1,1, \frac{1}{\Gamma_{sr}\delta^2_{sr}}
)_{k_1},(0,1,0) \end{array} \right. \right)
\end{align}
where in step $(a)$ we define $Y_{sd}=\prod_{i=1}^{k_1}(1+\Gamma_{sd}g_{sd,i})$. Since $\Gamma_{uv}g_{uv,i}\sim E(1/\Gamma_{uv}\delta^2_{uv})$. Therefore in step $(b)$ by using Corollary \ref{Lemma2}, $Y_{sd}$ is random variable define by the product of  $k_1$ shifted exponential random variable $H_i=(1+\Gamma_{sd}g_{sd,i})$  with the same parameter $1/\Gamma_{sd}\delta^2_{sd}$ and shifted parameter $\alpha=1$, i.e $H_i\sim ShE(1/\Gamma_{sd}\delta^2_{sd},1)$. Finally in step $(c)$, we use cdf given in (\ref{Product_cdf}) we get (\ref{Qsdk11}). Similarly, for $Q_{sr,k_1}$ is given by (\ref{Qsdk22}). In the case that only the relay decodes the transmitted packet after $k_1$ transmissions from the source and the relay sends $k_2$ times to the destination. The probability that the destination fails to decode the transmitted packet correctly after $k_1$ transmissions from the source and $k_2$ transmissions from the relay is
\begin{align}
Q_{srd,k_1;k_2}=&\mathbb{P}\Bigg\{\sum_{i=1}^{k_1}\log_2\left(1+\Gamma_{sd}g_{sd,i}\right)+\sum_{i=1}^{k_2}\log_2\left(1+\Gamma_{rd}g_{rd,i}\right)\leq R\Bigg\}\nonumber\\
               =&\mathbb{P}\Bigg\{
  \log_2\Bigg( \prod_{i=1}^{k_1}\left(1+\Gamma_{sd}g_{sd,i}\right)\prod_{i=1}^{k_2}(1+\Gamma_{rd}g_{rd,i})\Bigg)\leq R
  \Bigg\}\nonumber\\
  \stackrel{(a)}{=}&\mathbb{P}\left\{ Y_{sd}Y_{rd}\leq 2^R\right\}\nonumber\\
  \stackrel{(b)}{=}&\mathbb{P}\left\{Z\leq 2^R\right\}\nonumber\\
  \stackrel{(c)}{=}&\int_{0}^{2^R}f_Z(t) dt\nonumber\\
  \stackrel{(d)}{=}&F_Z(2^R)\nonumber\\ \label{Qsrdk11k22}
    =&e^{\frac{k_1}{\Gamma_{sd}\delta^2_{sd}}+\frac{k_2}{\Gamma_{rd}\delta^2_{rd}}}
   \large{H}^{k_1+k_2,1}_{1,k_1+k_2+1}
   \left( \left(\frac{1}{\Gamma_{sd}\delta^2_{sd}}\right)^{k_1} \left(\frac{1}{\Gamma_{rd}\delta^2_{rd}}\right)^{k_2}2^R \left|
   \footnotesize
   \begin{array}{cc}(1,1,0) \\ \Psi_1\cdots \Psi_{k_1},\Phi_{1}\cdots \Phi_{k_2},(0,1,0) \end{array} \right. \right)
\end{align}
where in step $(a)$,  we define $Y_{sd}=\prod_{i=1}^{k_1}(1+\Gamma_{sd}g_{sd,i})$ and $Y_{rd}=\prod_{i=1}^{k_2}(1+\Gamma_{rd}g_{rd,i})$ and in step $(b)$ we have $Z=Y_{sd}Y_{rd}$. Using Corollary \ref{Lemma2}, it is clear that $Z=Y_{sd}Y_{rd}$ is the product of two random variables where each one is defined by the product of $k_1$ and $k_2$ shifted exponential random variables, i.e $Y_{sd}$ is the product of $k_1$ of shifted random variables
$(1+\Gamma_{sd}g_{sd,i})\sim ShE(\frac{1}{\Gamma_{sd}\delta_{sd}^2},1)$ and $Y_{rd}$ is the product of $k_2$ of shifted random variables
$(1+\Gamma_{rd}g_{rd,i})\sim ShE(\frac{1}{\Gamma_{rd}\delta_{rd}^2},1)$. Therefore in step $(c)$, we use (\ref{pdf_Prod_shift_Diff}) and in step $(d)$ we use (\ref{cdf_Prod_shift_Diff}) we get (\ref{Qsrdk11k22}) and we have
\begin{align}
  \Psi_l=&(1,1,\frac{1}{\Gamma_{sd}\delta^2_{sd}}) \text{ for } l=1,\cdots,k_1,  \\
  \Phi_l=&(1,1,\frac{1}{\Gamma_{rd}\delta^2_{rd}}) \text{ for } l=1,\cdots,k_2.
\end{align}

In the case when channels are symmetric, i.e the same transmit power $P$, noise variance $N$ and channel parameter $1/\delta^2$ we use cdf given in (\ref{Product_cdf}) to get $Q_{sd,k_1}$ and $Q_{srd,k_1;k_2}$ which are given by (\ref{Qsdsrsrd1}) and (\ref{Qsdsrsrd2}), respectively and $Q_{sr,k_1}=Q_{sd,k_1}$.
\begin{align}\label{Qsdsrsrd1}
 Q_{sd,k_1} &=e^{\frac{k_1N}{P\delta^2}}\times\large{H}^{k_1,1}_{1,k_1+1}\left( \left(\frac{N}{P\delta^2}
\right)^{k_1}2^R \left| \begin{array}{cc}(1,1,0) \\ (1,1, \frac{N}{P\delta^2}
)_1,\cdots ,(1,1, \frac{N}{P\delta^2}
)_{k_1},(0,1,0) \end{array} \right. \right)\\
Q_{srd,k_1;k_2}&= e^{\frac{N(k_1+K_2)}{P\delta^2}} \nonumber\\ \label{Qsdsrsrd2}
  \times&
   \large{H}^{k_1+k_2,1}_{1,k_1+k_2+1}
   \left( \left(\frac{N}{P\delta^2}\right)^{k_1+k_2}2^R \left| \begin{array}{cc}(1,1,0) \\ (1,1,\frac{N}{P\delta^2})_1\cdots (1,1, \frac{N}{P\delta^2})_{k_1+k_2},(0,1,0) \end{array} \right. \right)
\end{align}
\section{Simulations and Results}
\subsection{Outage Probabilities}
In this section, the analytical expressions for the outage probability versus SNR $\Gamma$, for both HARQ-RR and IR
are evaluated numerically and validated by Monte Carlo simulations in Fig \ref{Fig_Outage} (a) for RR and Fig \ref{Fig_Outage} (b) for IR. Both the source and the relay are assumed to transmit with the same power equals one.
Simulation and numerical results are in perfect agreement. Moreover, as expected, increasing the number of retransmissions improves the
outage probability. For the case of asymmetric channels, we can see that when the SNR of the source-destination or the relay-destination channels increases the performance improves and vice versa. From Fig \ref{Fig_Outage} (b), we can see that IR outperforms RR and gives better performance.
\begin{figure*}[t]
\minipage{0.5\textwidth}
  \includegraphics[width=1\linewidth]{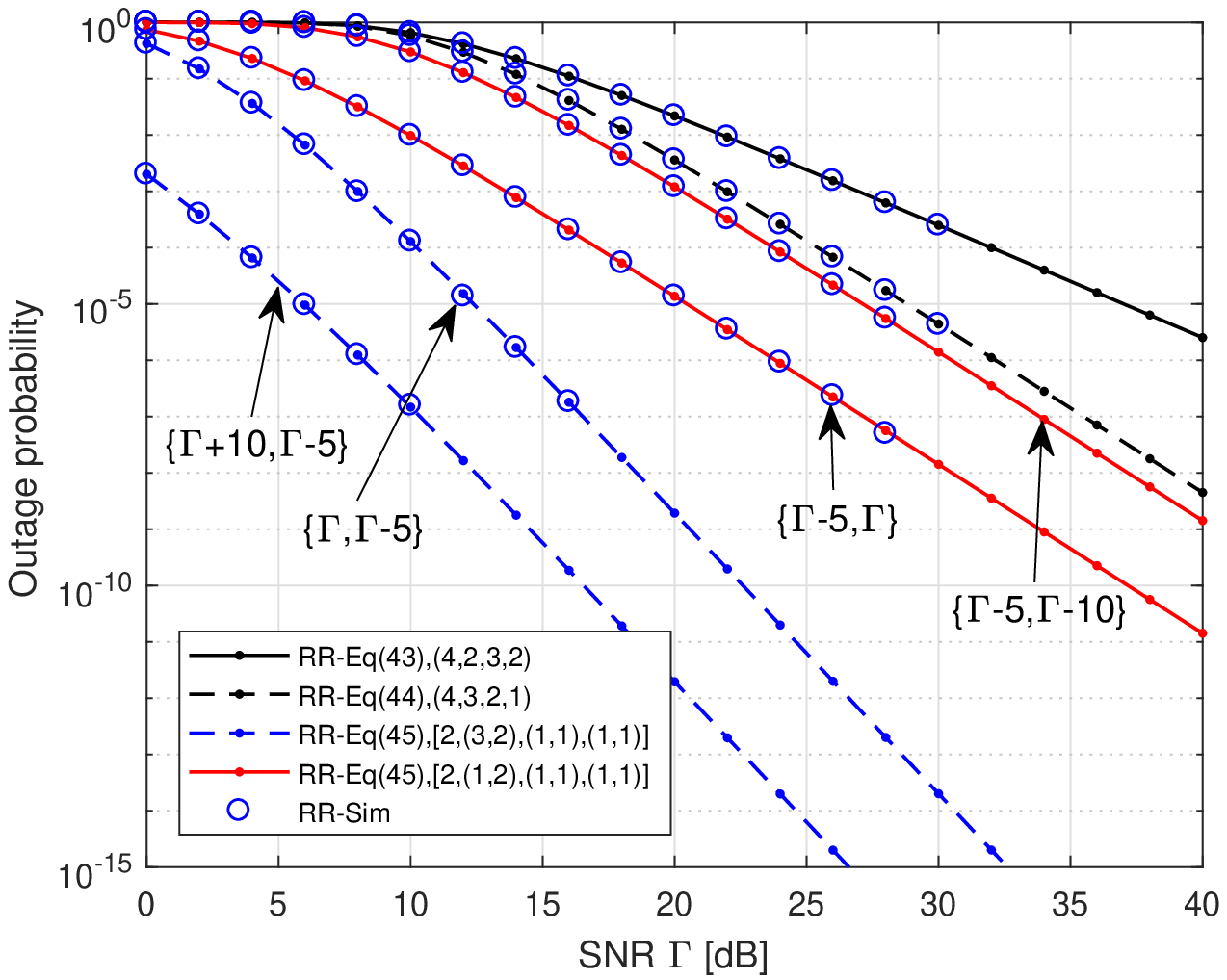}
\caption*{(a)}\label{Fig_outage_RR}
\endminipage
\hfill
\minipage{0.5\textwidth}%
  \includegraphics[width=1\linewidth]{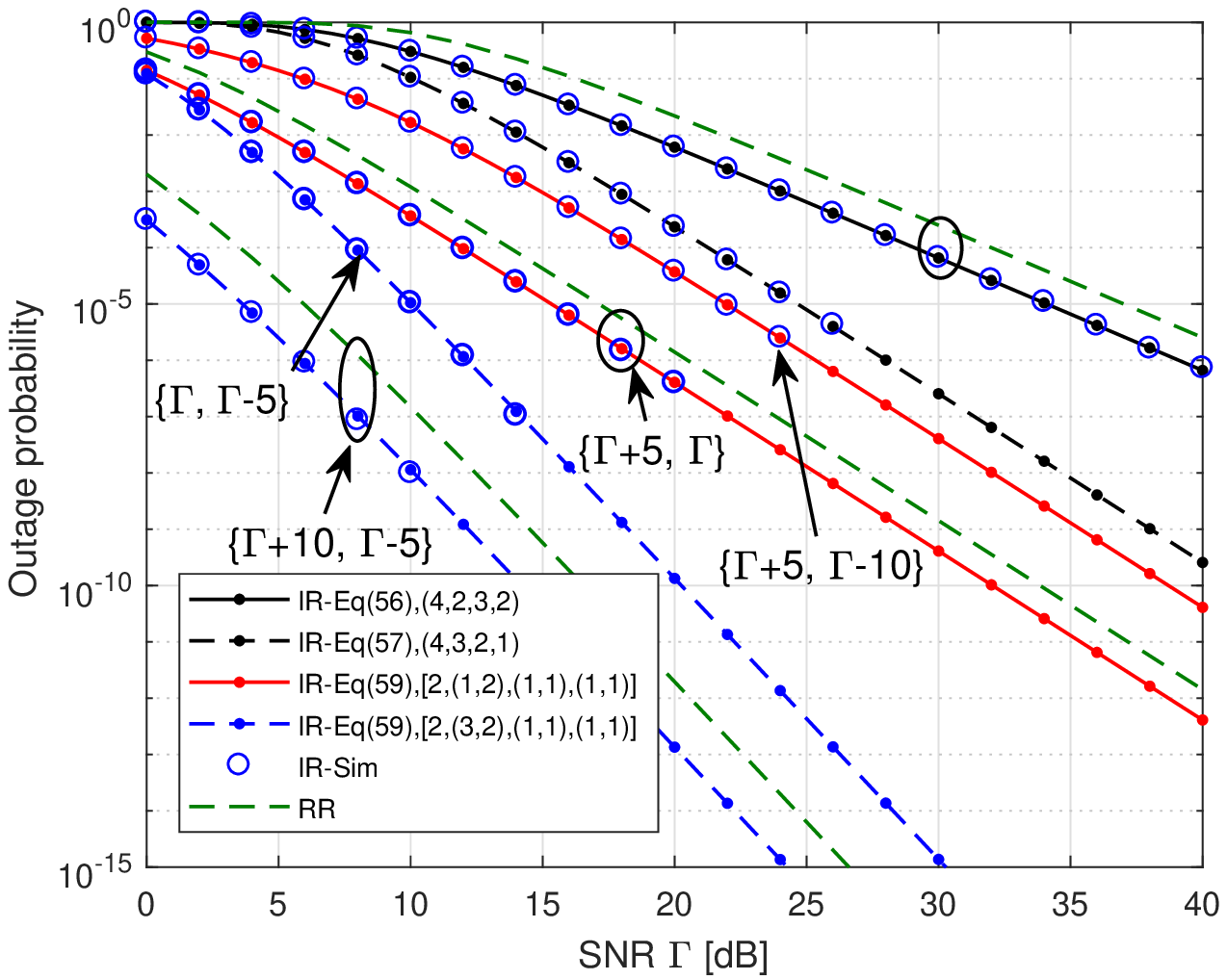}
\caption*{(b)}\label{Fig_outage_IR}
\endminipage
\caption{ The outage probability vs. $\Gamma$, where $(R,k_1,N_{uv},\delta^2_{uv})$ for $u\in\{s\}$ and $v\in\{r,d\}$, $\left[R,(k_1,k_2),(\delta^2_{sd},\delta^2_{rd}),(N_{sd},N_{rd})\right]$ and $\left\{\Gamma_{sd},\Gamma_{rd}\right\}$: (a) HARQ-RR, (b) HARQ-IR.}\label{Fig_Outage}
\end{figure*}
\subsection{Effective Capacity using HARQ-RR}
\begin{figure*}[t]
\minipage{0.5\textwidth}
  \includegraphics[width=1\linewidth]{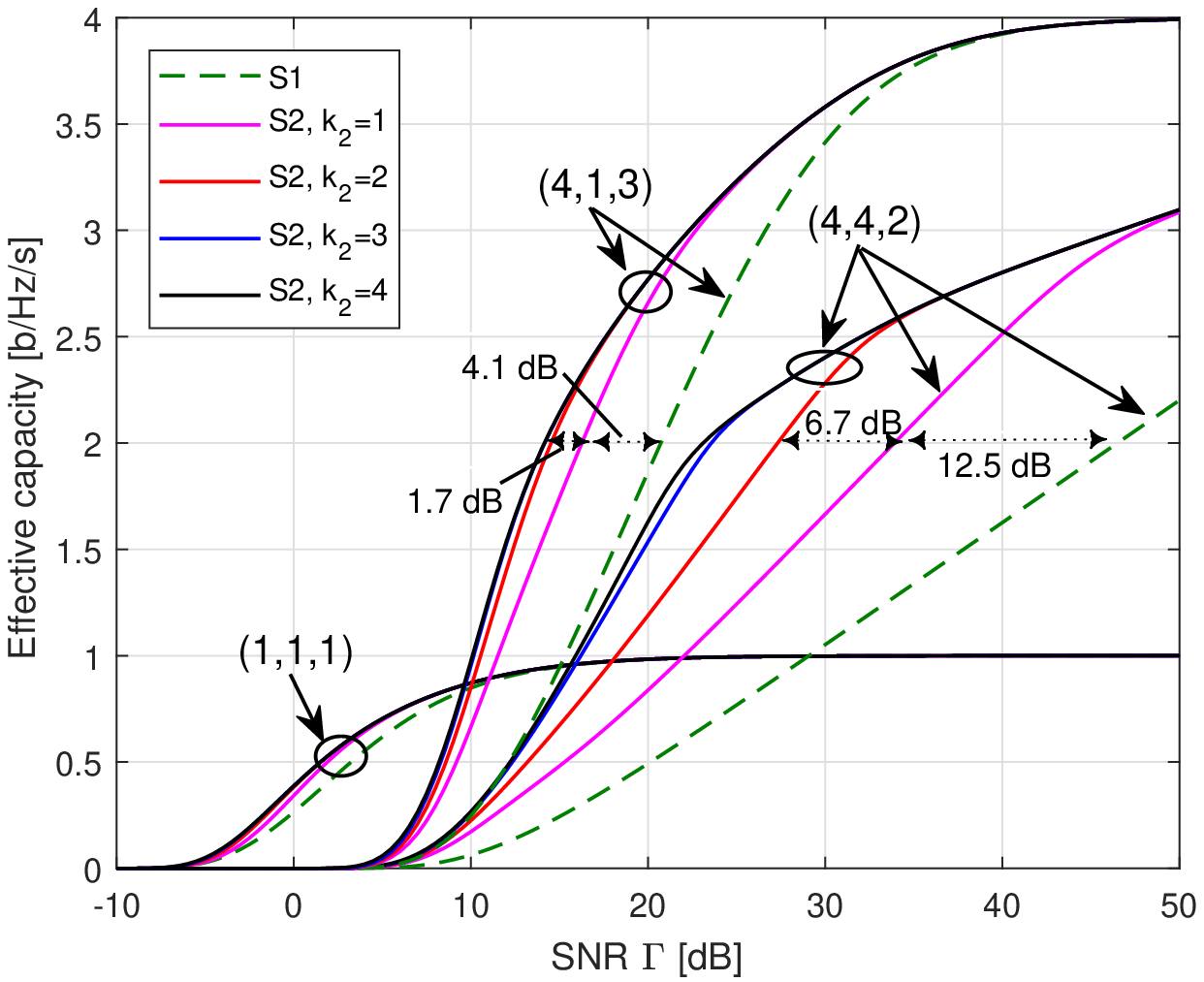}
\caption*{(a)}\label{Fig_EC_vs_SNR1}
\endminipage
\hfill
\minipage{0.5\textwidth}%
  \includegraphics[width=1\linewidth]{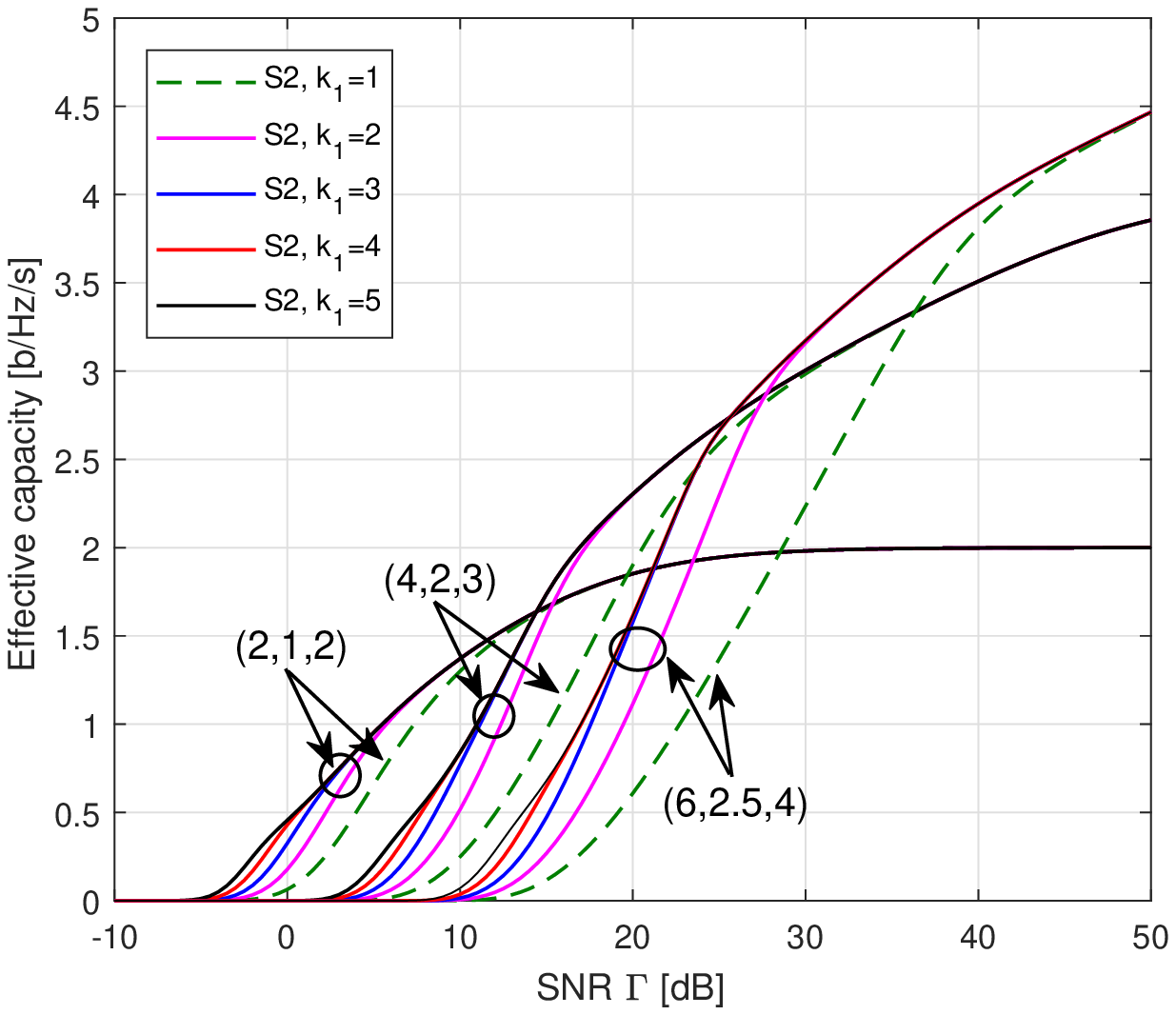}
\caption*{(b)}\label{Fig_EC_vs_SNR2}
\endminipage
\caption{Effective capacity vs. $\Gamma$ over symmetric channels: (a) $k_2$ number of relay transmissions and $(x, y, z)$ mean $(R,\theta, k_1)$, (b) $k_1$ number of source transmissions and $(R,\theta, k_2)$.}\label{EC_vs_SNR_RR}
\end{figure*}
In Fig \ref{EC_vs_SNR_RR} we plot the effective capacity vs SNR $\Gamma$ using both strategies, strategy I (S1)
and strategy II (S2), where both, the relay and the destination use HARQ-RR over symmetric channels. The transmit power $P$ and the noise power are
assumed to be one i.e, $P=N=\delta^2=1$.
For Fig \ref{EC_vs_SNR_RR} (a), we fix the number of transmissions at the source ($k_1$) and vary the number of transmissions at
the relay ($k_2$). We can see that, the effective capacity increases when $\Gamma$ increases and is close to the rate $R$ at high SNR $\Gamma$. Also, we have
1) for S2, when we set $k_1=1$ and $k_2=1$, the destination combines two packets, we can see that the effective capacity significantly
improves compared to S1 when $\theta$ and R increase,
2) when we increase the number of relay transmissions ($k_2\in\{1,2,3,4\}$), we can see that the effective capacity keeps improving up to four transmissions, negligible effective capacity gains, so $k_2=2$ may be sufficient for small $\theta$ and R, and $k_2=4$ may be sufficient for high $\theta$. Also, when $\theta$ increases the effective capacity gains increase when $k_2$ increases.
 Example when $R=4$ we can see that at $C=2$ b/Hz/s, when $\theta=1$ the gain between S1 and S2 with $k_1=k_2=1$, is $4.1$ dB and this gain increase to $12.5$ dB when $\theta=4$.
 Similarly, when we add one other transmission the gain between S2 with $k_1=k_2=1$ and $k_1=k_2=2$ is $1.7$ dB when $\theta=1$ and is $6.7$ dB when $\theta=4$.
In Fig \ref{EC_vs_SNR_RR} (b) we fix the number of transmissions at the relay $k_2$, and vary the number of transmissions by
the source $k_1$. When we increase the number of source transmission ($k_1\in\{1,2,3,4,5\}$), we can see that the effective capacity keeps improving especially at a low value of the effective capacity. However, in the middle value of the effective capacity,
 after four transmissions, we have negligible effective capacity gains, so $k_1=4$ may be sufficient.
\begin{figure*}[t]
\minipage{0.5\textwidth}
  \includegraphics[width=1\linewidth]{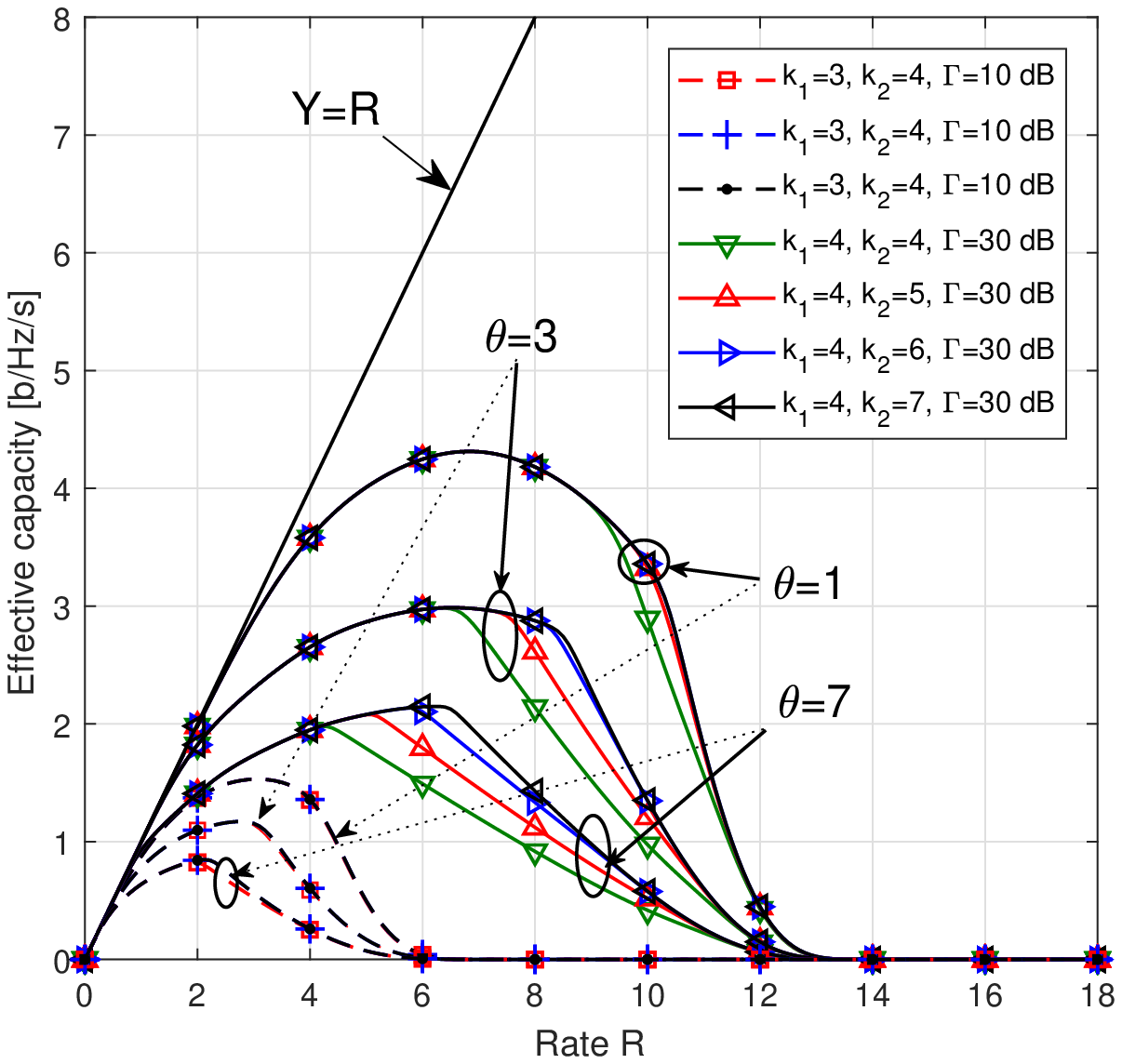}
\caption*{(a)}\label{Fig_EC_vs_R1}
\endminipage
\hfill
\minipage{0.5\textwidth}%
  \includegraphics[width=1\linewidth]{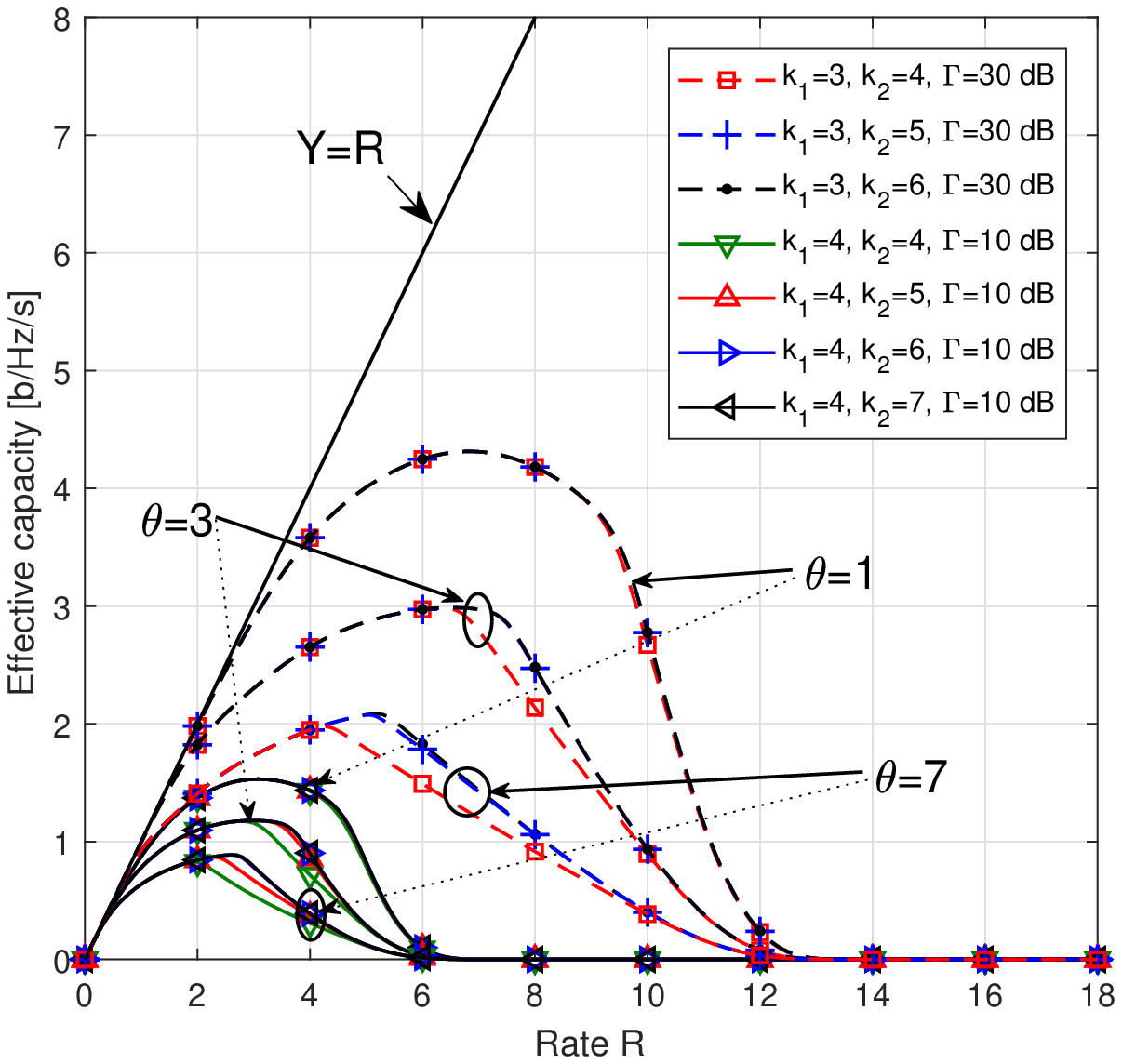}
\caption*{(b)}\label{Fig_EC_vs_R2}
\endminipage
\caption{Effective capacity vs. $R$ over symmetric channels where $P=N=\delta^2=1$: (a) $(k_1, k_2, \Gamma)\in\{(3, 4, 10)(4, 4, 30)(4, 5, 30)(4, 6, 30)(4, 7, 30\}$, (b) $(k_1, k_2, \Gamma)\in\{(3, 4, 30)(3, 5, 30)(3, 6, 30)(4, 4, 10)(4, 5, 10)(4, 6, 10)(4, 7, 10) \}$ .}\label{EC_vs_R}
\end{figure*}
In Fig \ref{EC_vs_R} we plot the effective capacity vs rate $R$ for different values of SNR $\Gamma$ ($\Gamma=10$ dB and $\Gamma=30$ dB ) and $\theta$ ($\theta\in\{1,3,7\}$).
In Fig \ref{EC_vs_R} (a) at low and high values of SNR we can see that, when the QoS Exponent $\theta$ is small, we can see that
after four transmissions from the relay, we have negligible effective capacity gains in all ranges of rate $R$ values.
When the QoS Exponent $\theta$ increases, we can see that the effective capacity can be improved when we add more transmissions
by the relay up to seven, i.e $k_2\in\{4,5,6,7\}$ but only in the middle range of rate $R$ and high SNR $\theta$ value. At a low SNR value, 
we can see that negligible effective capacity gains are obtained when more transmissions are added by the relay.
In Fig \ref{EC_vs_R} (b) we can observe the same properties.
\begin{figure*}[t]
\minipage{0.33\textwidth}
  \includegraphics[width=1\linewidth]{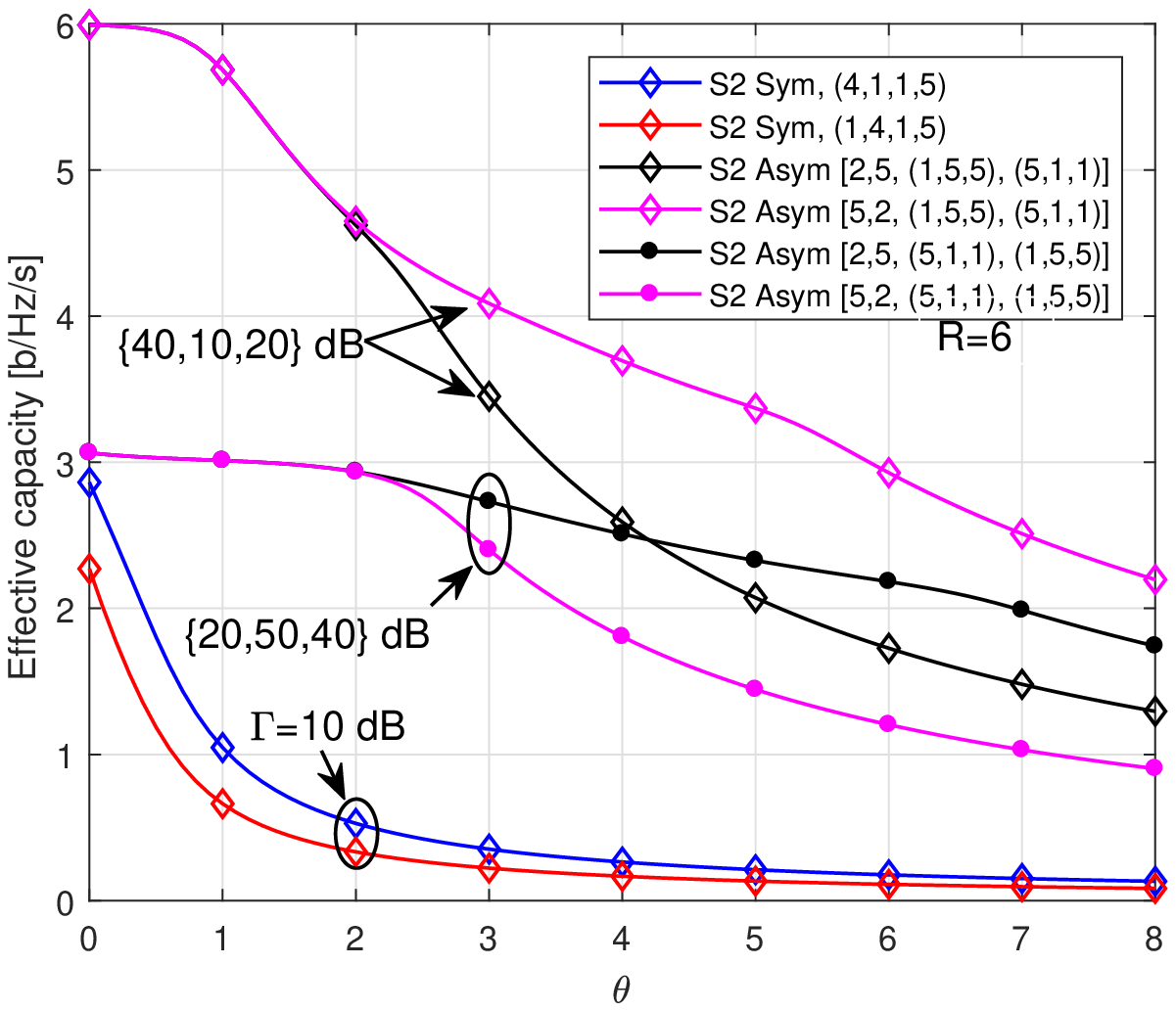}
\caption*{(a)}\label{EC_vs_Theta1}
\endminipage
\hfill
\minipage{0.33\textwidth}%
  \includegraphics[width=1\linewidth]{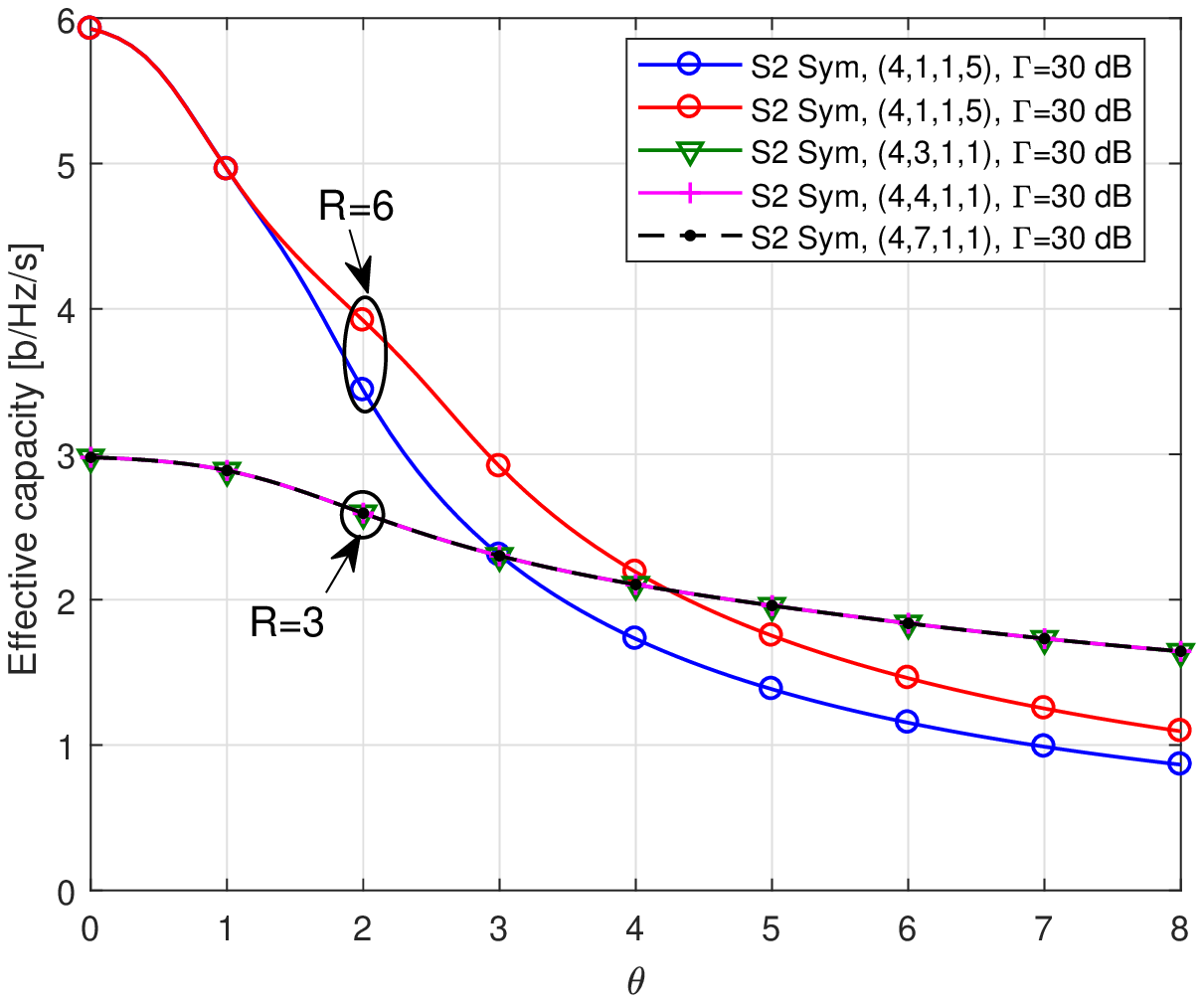}
\caption*{(b)}\label{EC_vs_Theta2}
\endminipage
\hfill
\minipage{0.33\textwidth}%
  \includegraphics[width=1\linewidth]{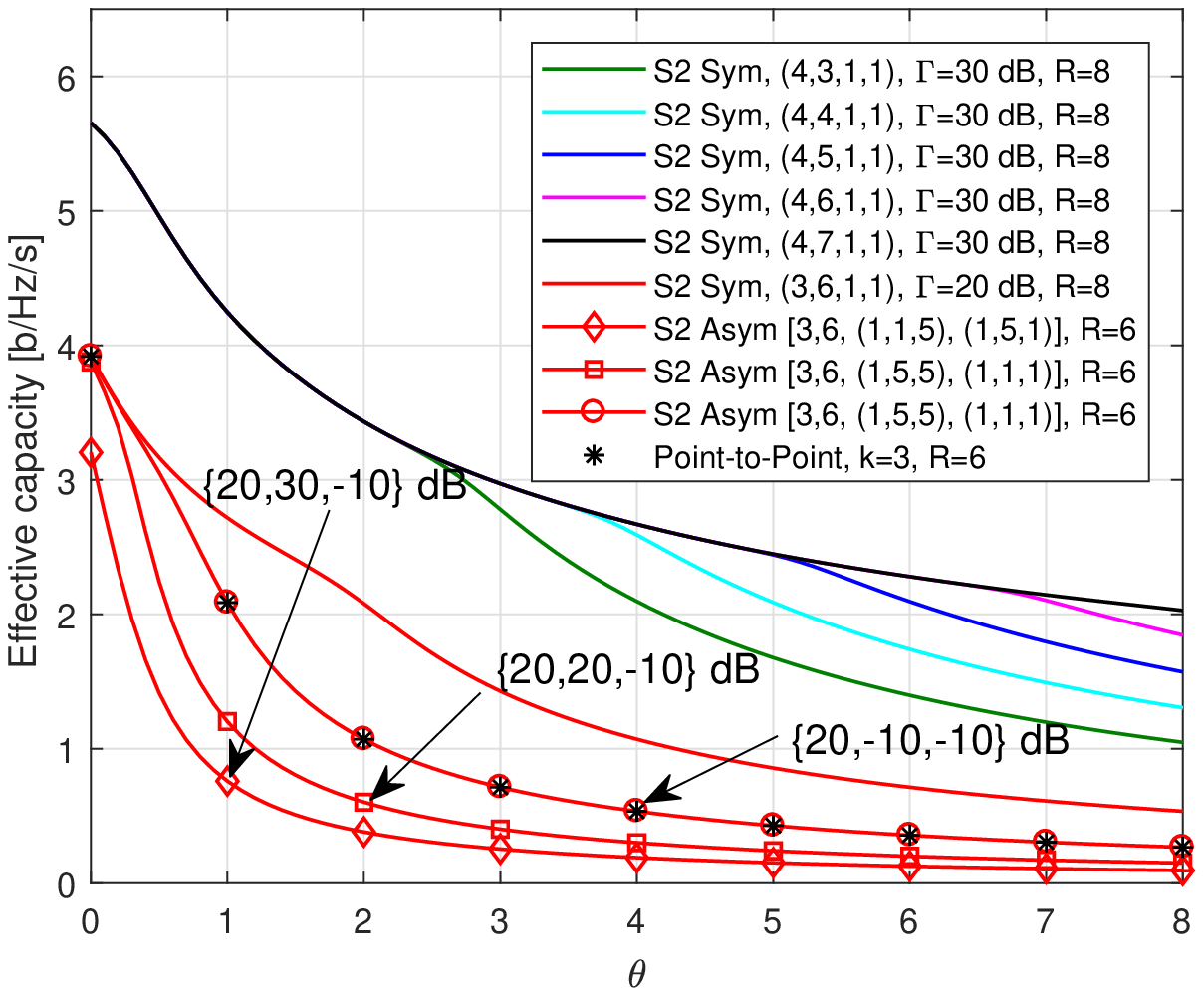}
\caption*{(c)}\label{EC_vs_Theta3}
\endminipage
\caption{Effective capacity vs. $\theta$ where $(k_1,k_2,N,\delta^2)$ and $[k_1,k_2,(N_{sd},N_{sr},N_{rd}),(\delta_{sd}^2,\delta_{sr}^2,\delta_{rd}^2)]$ and $\{\Gamma_{sd}, \Gamma_{sr}, \Gamma_{rd}\}$ dB  (a)  (b)  (c)}\label{EC_vs_THETA}
\end{figure*}
In Fig \ref{EC_vs_THETA}, we plot the effective capacity vs $\theta$ for different values of SNR and rate $R$.
In Fig \ref{EC_vs_THETA} (a) at low SNR values, we can see that when we have more transmissions at the source, the effective capacity is better compared to the case when we increase the number of transmissions at the relay.
Also, when the source-destination channel is better than the source relay and relay destination we can see that when the number of source transmissions is high gives better performance when we have a high number of transmissions at the relay. However, when the source relay and relay channels are better compared to the direct channel we can see that when the number of retransmissions by the relay is high gives better performance compared to the case when the number of source transmission is high.

In Fig \ref{EC_vs_THETA} (b) at high SNR values, we can see that we have more transmissions at the relay, the effective capacity is better compared to when we increase the number of transmissions at the source.
Also for small values of rate $R$, when we increase the number of transmissions by the relay more than four transmissions,
we have negligible effective capacity gains in all ranges of $\theta$.
In Fig \ref{EC_vs_THETA} (c), we can see that for high values of rate $R$, when we increase the number of transmissions by the relay the effective capacity improves when $\theta$ increases. Also, in the case when the relay destination channel is too bad we can see that the more the relay does not participate, the effective capacity improves, and the Maximus values can archive when the relay does not participate is a point to point $k_1$-HARQ.
\subsection{Effective Capacity using HARQ-IR}
\begin{figure*}[t]
\minipage{0.5\textwidth}
  \includegraphics[width=1\linewidth]{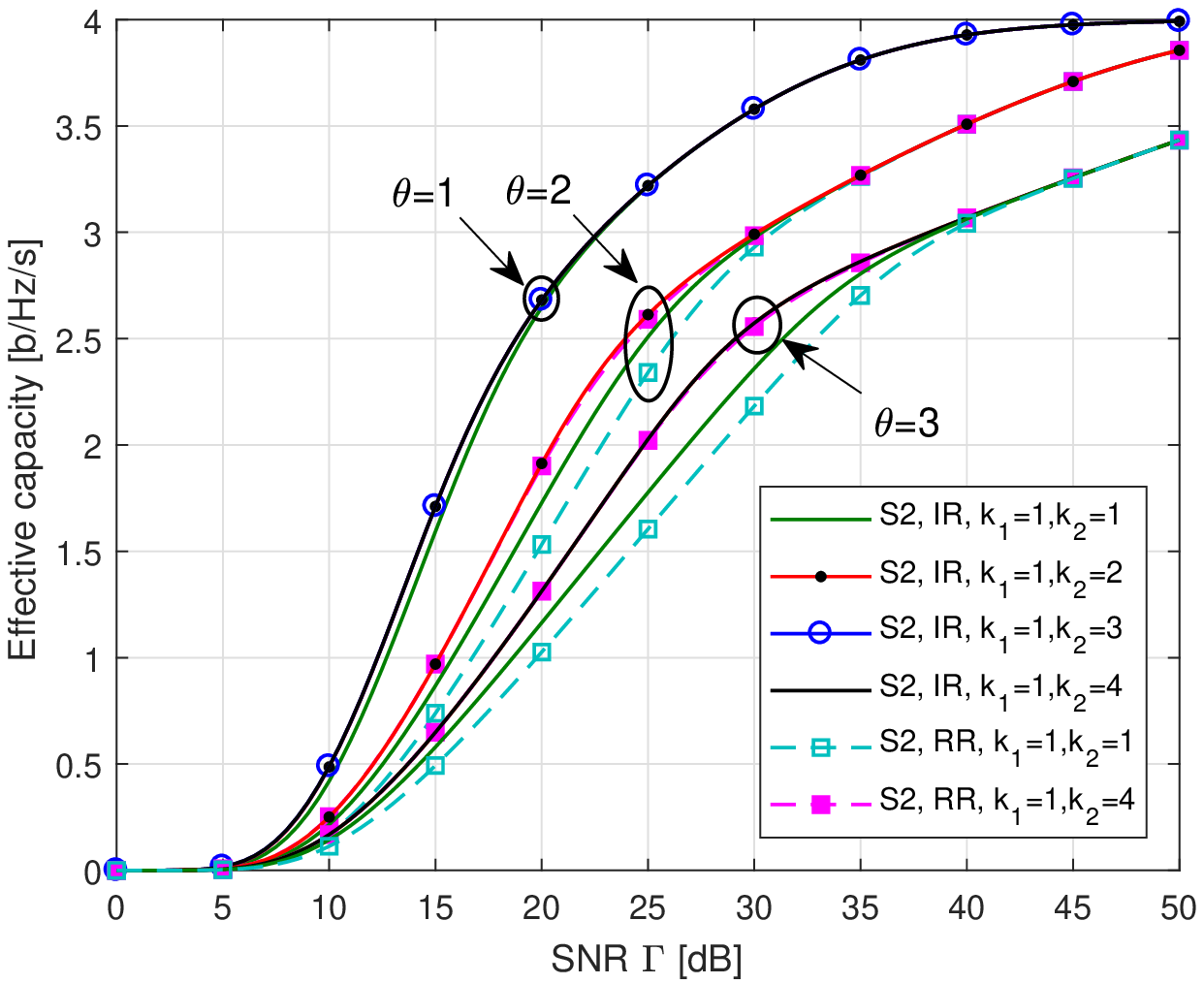}
\caption*{(a)}\label{EC_vs_SNR_IR_1}
\endminipage
\hfill
\minipage{0.5\textwidth}%
  \includegraphics[width=1\linewidth]{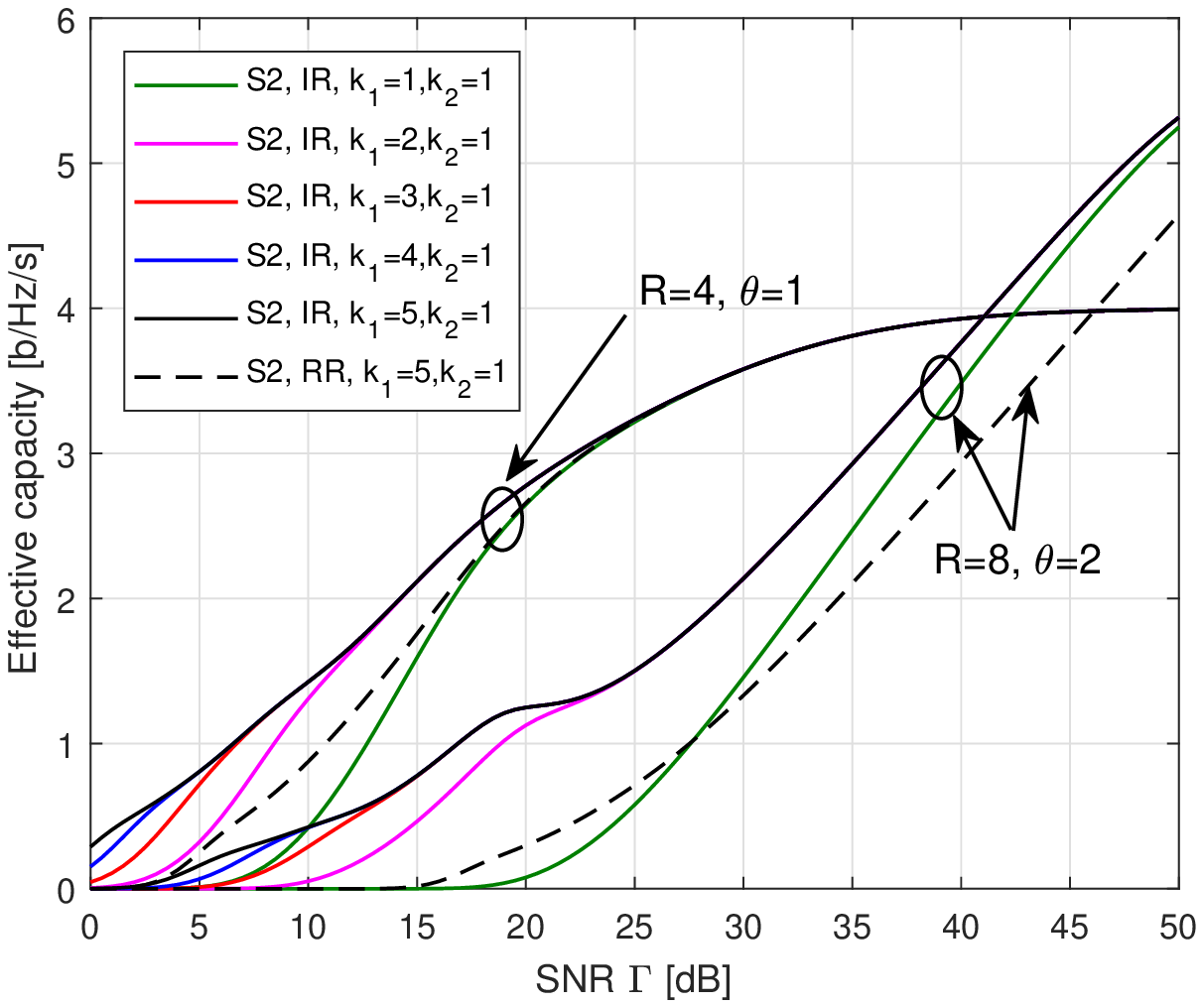}
\caption*{(b)}\label{EC_vs_SNR_IR_1}
\endminipage
\caption{Effective capacity vs. $\Gamma$ over symmetric channels: (a) rate $R=4$ b/Hz/s, $k_1=1$ and $k_2\in\{1,2,3,4\}$, (b) $k_1\in\{1,2,3,4\}$ and $k_2=1$.}\label{EC_vs_SNR_IR}
\end{figure*}
In Fig \ref{EC_vs_SNR_IR}, we plot the effective capacity vs SNR $\Gamma$ for strategy II (S2),
 where both relay and destination use HARQ-IR over symmetric channels, with $P=N=\delta^2=1$.
In Fig \ref{EC_vs_SNR_IR} (a), we fix $k_1$ and we vary $k_2$, as we can see that, for small values of rate $R$ and
$\theta$, only two transmissions i.e, $k_2=2$ are sufficient and after two transmissions we have negligible effective capacity gains. Also, in the case of $k_1=k_1=1$, we can see that HARQ-IR is better than HARQ-RR.
 In Fig \ref{EC_vs_SNR_IR} (b), we fix $k_2$  and we vary $k_1$, as we can see that,
 increasing $k_1$, the effective capacity keeps improving specially at low values of the effective capacity. However, in the middle value of the effective capacity, after four transmissions, we have negligible effective capacity gains, so $k_1=4$ may be sufficient.
\begin{figure*}[t]
\minipage{0.5\textwidth}
  \includegraphics[width=1\linewidth]{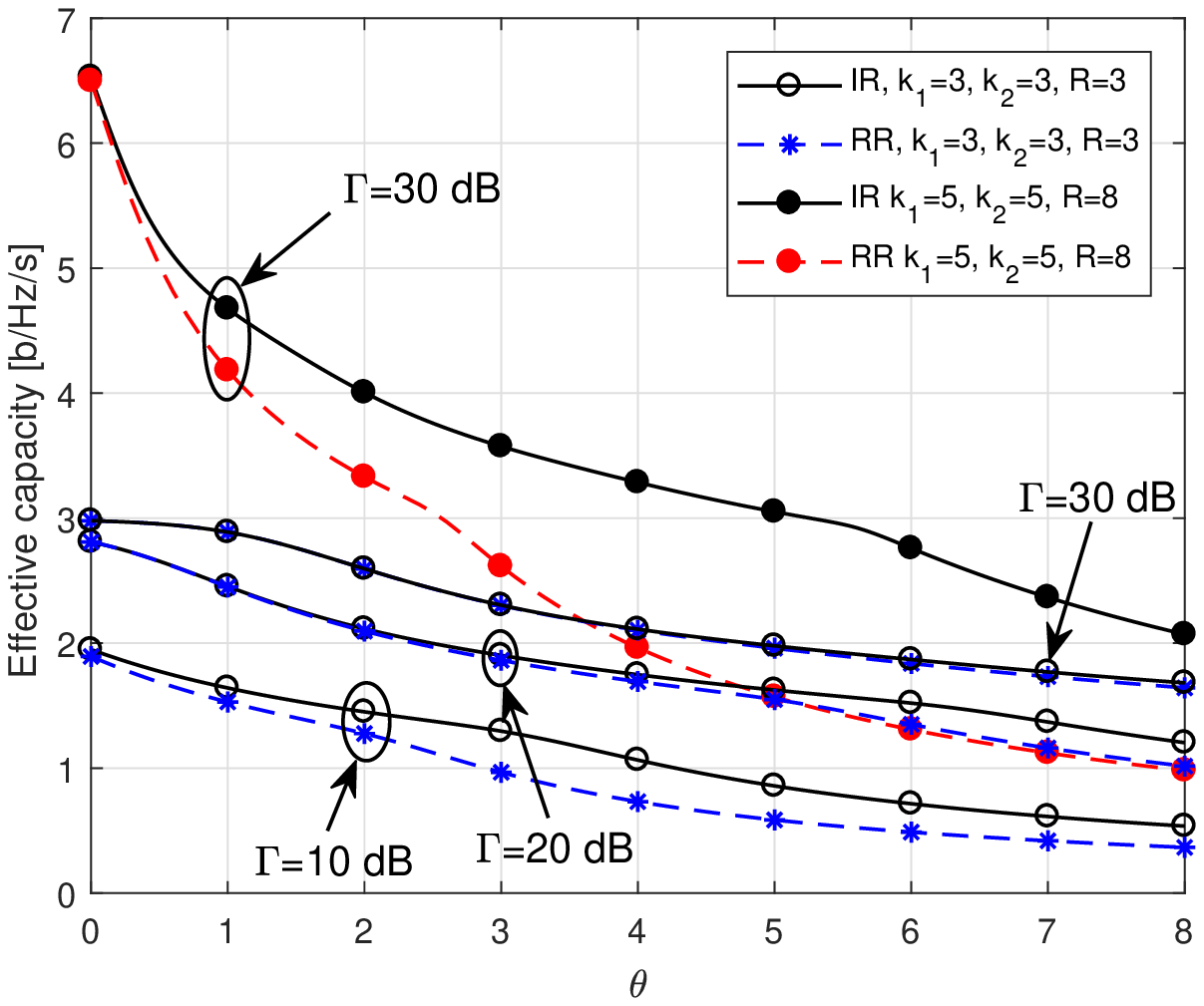}
\caption*{(a)}\label{EC_vs_theta_IR_1}
\endminipage
\hfill
\minipage{0.5\textwidth}%
  \includegraphics[width=1\linewidth]{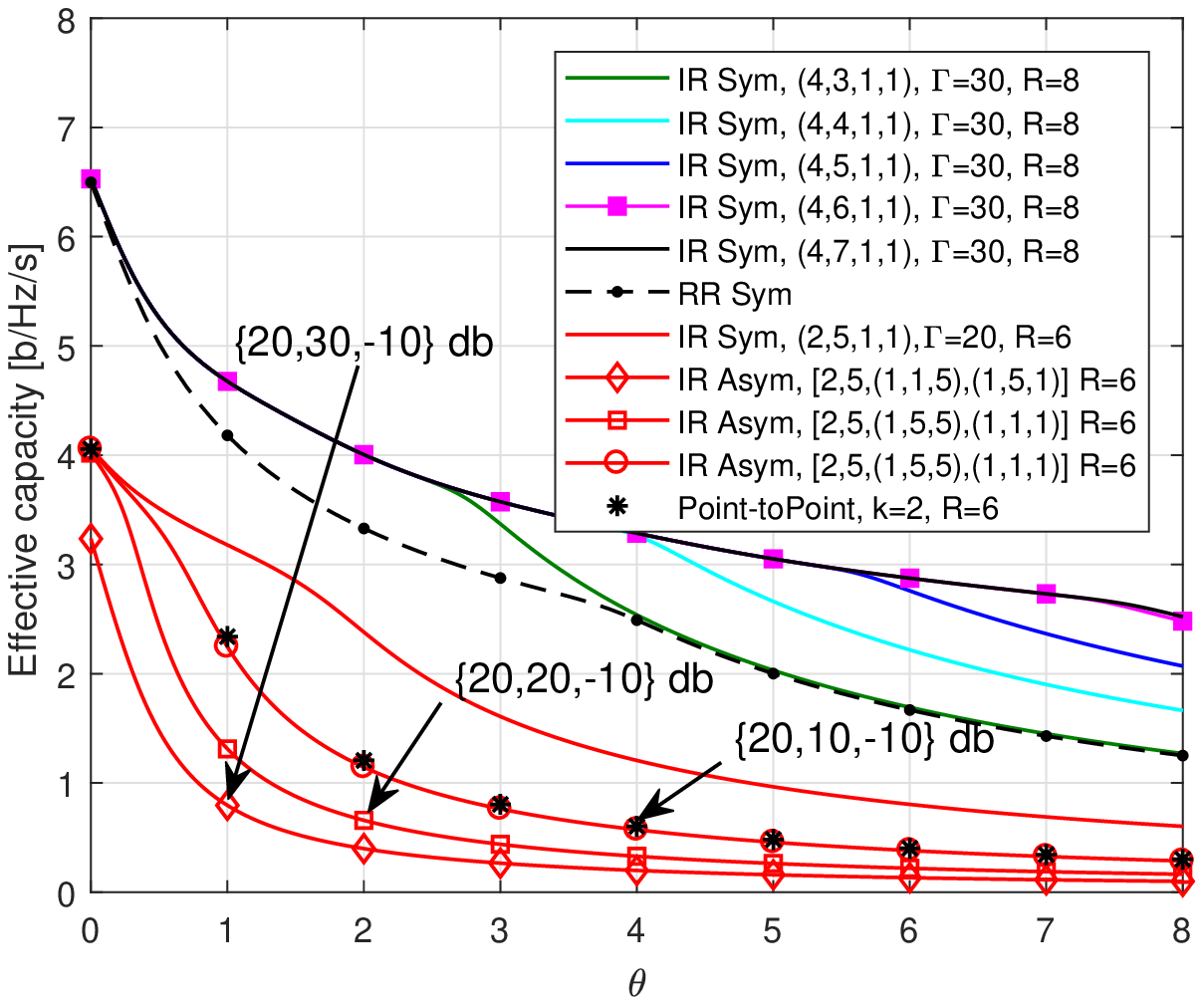}
\caption*{(b)}\label{EC_vs_theta_IR_2}
\endminipage
\caption{Effective capacity vs. $\theta$  where $(k_1,k_2,N,\delta^2)$ and $[k_1,k_2,(N_{sd},N_{sr},N_{rd}),(\delta_{sd}^2,\delta_{sr}^2,\delta_{rd}^2)]$ and $\{\Gamma_{sd}, \Gamma_{sr}, \Gamma_{rd}\}$ dB}\label{EC_vs_theta_IR}
\end{figure*}
In Fig \ref{EC_vs_theta_IR} we plot the effective capacity vs $\theta$ for different values of SNR $\Gamma$ and rate $R$.
In Fig \ref{EC_vs_theta_IR} (a) at low SNR and rate values, we can see that HARQ-IR is better than HARQ-RR specially when $\theta$ increases. The gain 
keeps reducing and becomes negligible at high SNR in all the range of $\theta$.
At high rate values we can see that HARQ-IR is better than HARQ-RR also at high SNR specially when $\theta$ increases, contrary to low value of $R$.
In Fig \ref{EC_vs_theta_IR} (b) similar to HARQ-RR we can see that for high values of rate $R$, when the number of transmissions by the relay increases , the effective capacity improves when $\theta$ increases. Also, HARQ-IR is better than HARQ-RR. Also, in the case when the relay destination channel is too bad, we can see that more the relay does not participate, the effective capacity improves, and the maximus values can archive when the relay do not participate is point to point $k_1$-HARQ.
\section{Conclusion}
In this paper based on the recurrence relation approach, the effective capacity for CC has been analyzed for retransmission schemes. The derived effective capacity can be used for any distribution over asymmetric channels. Also, the derived expressions can be used for any number of transmissions assigned to the source and the relay and both relay and destination use packets combined. Also depending on the rate values, $\Gamma$ value and $\theta$ value and number of transmissions by the source and relay, we have the following results, for symmetric channels, when the $\Gamma$ is low,  it is better to increase the number of source transmissions, and vice versa. When we increase the number of the source transmissions the effective capacity improves only at low values of $\Gamma$. Also for low rate $R$ and $\theta$, after four transmissions at the relay the effective capacity gain is negligible. When $R$ and $\theta$, increase we can see that the effective capacity improves when we add more transmissions at the relay specially at the middle of $\Gamma$ region.
For asymmetric channels, when the channel between the relay and destination is too bad we have seen that, more the relay do not participate more the effective capacity improves, and the performance are upper bounded by the performance of the effective capacity for point to point using HARQ between source and destination. Also when we compare HARQ-IR with HARQ-IR, the improvements of HARQ-IR become significant when $R$ and $\theta$ increase.
\appendix
\normalsize
\subsection{Proof of Corollary 1}\label{App2}
From \cite{jasiulewicz_convolutions_2003}\cite{kadri_convolutions_2015}, the sum of $k_i$ exponential random variable with the same parameter $\mu_i$, is a random variable $Y$ has an Erlang distribution with parameter $k_i$ and $\mu_i$ i.e $Y\sim Erl(k_i,\mu_i)$ with pdf $f_{Y_i}(t)$ for $t>0$ given by
\begin{equation}
  f_{Y}(t)=\frac{\mu_i^{k_i}t^{k_i-1}}{(k_i-1)!}e^{-\mu_it}.
\end{equation}
The cdf of random variable $Y\sim Erl(k_i,\mu_i)$ is $F_Y(t)$ given by
\begin{align}
F_{Y}(t)&=\int_{0}^{t}\frac{\mu_i^{k_i}x^{k_i-1}}{(k_i-1)!}e^{-\mu_ix}dx \\
          &\stackrel{(a)}{=}\frac{1}{(k_i-1)!}\gamma(k_i,t\mu_i)
\end{align}
where in step $(a)$ we use \cite[3.351.1]{gradshtein_table_2015}.
For random variable $Z$ define by the sum of two Erlang random variable $Y_1\sim Erl(k_1,\mu_1)$ and $Y_2\sim Erl(k_2,\mu_2)$, the pdf is $f_Z(t)$ is given by
\begin{align}
  f_Z(t)=&\sum_{i=1}^{2}\mu_i^{k_i}e^{-\mu_it}\sum_{j=1}^{k_i}\frac{(-1)^{k_i-j}}{(j-1)!}t^{j-1}\nonumber\\ \label{6}
   \times&\sum_{\begin{matrix}
  n_1+n_2=k_i-j \\
  n_i=0
\end{matrix}}^{}
\prod_{\begin{matrix}
  l=1 \\
  l\neq i
\end{matrix}}^{2}
\binom{k_l+n_l-1}{n_l}\frac
{\mu_l^{k_l}}
{(\mu_l-\mu_i)^{k_l+n_l}}.
\end{align}
The cdf of $Z$ is $F_Z(t)$ given by
\begin{align}
  F_Z(t)=&\int_{0}^{t}f_Z(x)dx\nonumber\\
        =&\sum_{i=1}^{2}\mu_i^{k_i}\sum_{j=1}^{k_i}\frac{(-1)^{k_i-j}}{(j-1)!}\int_{0}^{t}x^{j-1}e^{-\mu_ix}dx
   \times\sum_{\begin{matrix}
  n_1+n_2=k_i-j \nonumber\\
  n_i=0
\end{matrix}}^{}
\prod_{\begin{matrix}
  l=1 \nonumber\\
  l\neq i
\end{matrix}}^{2}
\binom{k_l+n_l-1}{n_l}\frac
{\mu_l^{k_l}}
{(\mu_l-\mu_i)^{k_l+n_l}}\\
=&\sum_{i=1}^{2}\mu_i^{k_i}\sum_{j=1}^{k_i}\frac{(-1)^{k_i-j}}{(j-1)!}
\mu_i^{-j}\gamma(j,t\mu_i)
   \times\sum_{\begin{matrix}
  n_1+n_2=k_i-j \nonumber\\
  n_i=0
\end{matrix}}^{}
\prod_{\begin{matrix}
  l=1 \nonumber\\
  l\neq i
\end{matrix}}^{2}
\binom{k_l+n_l-1}{n_l}\frac
{\mu_l^{k_l}}
{(\mu_l-\mu_i)^{k_l+n_l}}.
\end{align}
\subsection{Proof Corollary 2}\label{App3}
Let $H_i\sim ShE(\mu_i,\alpha_i)$. Let $Z_i$ be a random variable define by the product of $k_i$ shifted
exponential random variable $H_j$ with the same
parameter $\mu_i$ and the shifted parameter $\alpha_i$, i.e,
$H_j\sim ShE(\mu_i,\alpha_i)$ for $j\in\{1,2,\cdots,k_i\}$ and $Z_i$ define by
\begin{equation}
  Z_i=\prod_{j=1}^{k_i}H_j
\end{equation}
We denote $f_{Z_i}(z)$ the pdf of $Z_i$, then from \cite{yilmaz_product_2009} we have
\begin{equation}\label{MelProduct}
   \mathcal{M}_s\{f_{Z_i}(z)\}(s)=\prod_{j=1}^{k_i}\mathcal{M}_s\{f_{H_j}(z)\}(s)
\end{equation}
The Mellin transform of $H_j\sim ShE(\mu_i,\alpha_i)$ can be derived as
\begin{align}
  \mathcal{M}_s\{f_{H_j}(z)\}(s)&=\int_{\alpha_i}^{\infty}z^{s-1}f_{H_j}(z)dz\nonumber\\
 &=\int_{\alpha_i}^{\infty}z^{s-1}\mu_i e^{-(z-\alpha_i)\mu_i}dz\nonumber\\
 &=\mu_i e^{\alpha_i \mu_i} \int_{\alpha_i}^{\infty}z^{s-1} e^{-z\mu_i}dz\nonumber\\ \label{MelShiftE}
 &\stackrel{(a)}{=}\mu_i^{1-s} e^{\alpha_i \mu_i}\Gamma(s,\alpha_i \mu_i).
\end{align}
where in step $(a)$ we use \cite[3.351.2]{gradshtein_table_2015}. Thus, the Mellin transform of $Z_i$ can be determined
by substituting (\ref{MelShiftE}) into (\ref{MelProduct}) gives
\begin{align}
  \mathcal{M}_s\{f_{Z_i}(z)\}(s)&=\prod_{j=1}^{k_i}\mu_i^{1-s} e^{\alpha_i \mu_i}\Gamma(s,\alpha_i \mu_i)\\
                                &=\mu_i^{(1-s)k_i} e^{k_i\alpha_i \mu_i}\Gamma(s,\alpha_i \mu_i)^{k_i}
\end{align}
The PDF of $Z_i$ can be found using the inverse Mellin transform as \cite{yilmaz_product_2009}
\begin{align}
  f_{Z_i}(z)&=\frac{1}{2\pi i}\int_{\gamma-i\infty}^{\gamma+i\infty}\mathcal{M}_s\{f_{Z_i}(z)\}(s)z^{-s}ds\\
   &=\mu_i^{k_i}e^{k_i\alpha_i \mu_i}\frac{1}{2\pi i}\int_{\gamma-i\infty}^{\gamma+i\infty}\mu_i^{-sk_i} \Gamma(s,\alpha_i \mu_i)^{k_i}z^{-s}ds \\ \label{Amine1}
   &=\mu_i^{k_i}e^{k_i\alpha_i \mu_i}\frac{1}{2\pi i}\int_{\gamma-i\infty}^{\gamma+i\infty} \Gamma(s,\alpha_i \mu_i)^{k_i}\left(\mu_i^{k_i}z\right)^{-s}ds.
\end{align}
Finally, (\ref{Product_pdf}) is a direct consequence from (\ref{Amine1}). The cdf of $Z_i$ is $F_{Z_i}(z)$ can be deduced from the pdf $f_{Z_i}(z)$ as given in equations (\ref{Equation1})-(\ref{Equation2}) \cite{chelli_performance_2014}
\begin{figure*}[!t]
\begin{align}\label{Equation1}
  F_{Z_i}(z)&=\int_{0}^{z}f_{Z_i}(t)dt\\
   &=\mu_i^{k_i}e^{k_i\alpha_i \mu_i}\frac{1}{2\pi i}\int_{\gamma-i\infty}^{\gamma+i\infty} \Gamma(s,\alpha_i \mu_i)^{k_i}
   \left[\int_{0}^{z}\left(\mu_i^{k_i}t\right)^{-s}dt\right]ds\\
   &=e^{k_i\alpha_i \mu_i}\frac{1}{2\pi i}\int_{\gamma-i\infty}^{\gamma+i\infty}\frac{\Gamma(s,\alpha_i \mu_i)^{k_i}}{1-s}
   \left(\mu_i^{k_i}z\right)^{-(s-1)}ds\\ \label{Equation2}
&\stackrel{(a)}{=}e^{k_i\alpha_i \mu_i}\frac{1}{2\pi i}\int_{\gamma-i\infty}^{\gamma+i\infty}\frac{\Gamma(1+h,\alpha_i \mu_i)^{k_i}\Gamma(-h,0)}{\Gamma(1-h)}
   \left(\mu_i^{k_i}z\right)^{-(h)}dh
\end{align}
\hrulefill
\end{figure*}
where in step $(a)$ we use $1-s=\Gamma(2-s)/\Gamma(1-s)$, then the change of variable $h=s-1$ \cite{chelli_performance_2014}. Finally, (\ref{Product_cdf}) is a direct consequence from (\ref{Equation2}).
Let $Z_1=\prod_{n=1}^{k_1}H_n$ the product of $k_1$ shifted exponential random variables $H_n$ with the same parameter $\mu_1$  and the shifted paramete $\alpha_1$, i.e $H_n\sim ShE(\mu_1,\alpha_1), \forall n\in\{1,2,\cdots,k_1\}$. Similar Let $Z_2=\prod_{m=1}^{k_2}H_m$ where $H_m\sim ShE(\mu_2,\alpha_2), \forall m\in\{1,2,\cdots,k_2\}$. Let $Z$ be a random variable define by
\begin{equation}
  Z=Z_1\times Z_2=\prod_{n=1}^{k_1}H_n\prod_{m=1}^{k_2}H_m
\end{equation}
We denote $f_{Z}(z)$ the pdf of $Z$,
\begin{equation}\label{MelProduct}
   \mathcal{M}_s\{f_{Z}(z)\}(s)=\prod_{n=1}^{k_1}\mathcal{M}_s\{f_{H_n}(z)\}(s)\prod_{m=1}^{k_2}\mathcal{M}_s\{f_{H_m}(z)\}(s).
\end{equation}
The Mellin transform of $Z$ can be determined
\begin{align}
  \mathcal{M}_s\{f_{Z}(z)\}(s)&=\mu_1^{k_1} \mu_2^{k_2} e^{k_1\alpha_1 \mu_1+k_2\alpha_2 \mu_2}\Gamma(s,\alpha_1 \mu_1)^{k_1}
  \times \Gamma(s,\alpha_2 \mu_2)^{k_2}
  \left(\mu_1^{k_1} \mu_2^{k_2}\right)^{-s}
\end{align}
The PDF of $Z$ can be found using the inverse Mellin transform as
\begin{align}\label{Equation3}
  f_{Z}(z)&=\frac{1}{2\pi i}\int_{\gamma-i\infty}^{\gamma+i\infty}\mathcal{M}_s\{f_{Z_i}(z)\}(s)z^{-s}ds\\
   &=\mu_1^{k_1} \mu_2^{k_2} e^{k_1\alpha_1 \mu_1+k_2\alpha_2 \mu_2}\frac{1}{2\pi i}\int_{\gamma-i\infty}^{\gamma+i\infty}
   \Gamma(s,\alpha_1 \mu_1)^{k_1}\times\Gamma(s,\alpha_2 \mu_2)^{k_2}
  \left(\mu_1^{k_1} \mu_2^{k_2}z\right)^{-s}ds
\end{align}
Finally, (\ref{pdf_Prod_shift_Diff}) is a direct consequence from (\ref{Equation3}).The cdf of $Z$ is $F_{Z}(z)$ can be deduced from the pdf $f_{Z}(z)$ as given in equations (\ref{Equation4})-(\ref{Equation5}), Finally, (\ref{cdf_Prod_shift_Diff}) is a direct consequence from (\ref{Equation5})
\begin{align}\label{Equation4}
  F_{Z}(z)&=\int_{0}^{z}f_{Z}(t)dt\\ \label{Equation5}
   &=\mu_1^{k_1} \mu_2^{k_2} e^{k_1\alpha_1 \mu_1+k_2\alpha_2 \mu_2}
   \frac{1}{2\pi i}\int_{\gamma-i\infty}^{\gamma+i\infty}
    \Gamma(s,\alpha_1 \mu_1)^{k_1}\Gamma(s,\alpha_2 \mu_2)^{k_2}
  \left[\int_{0}^{z}\left(\mu_1^{k_1} \mu_2^{k_2}t\right)^{-s}dt\right]ds
\end{align}
\footnotesize
\bibliographystyle{ieeetr}
\bibliography{Ref}

\end{document}